\documentclass[twocolumn]{article}

\usepackage[margin=0.8in]{geometry}

\usepackage{lmodern}
\usepackage[utf8]{inputenc}
\usepackage[T1]{fontenc}

\usepackage[
backend=biber,
style=alphabetic,
sorting=ynt
]{biblatex}
\usepackage{authblk}
\usepackage{blindtext}
\usepackage{graphicx}
\usepackage{url}
\usepackage[section]{placeins}
\usepackage{tikz}
\usetikzlibrary{quantikz2}
\usetikzlibrary{calc}
\usetikzlibrary{math}

\graphicspath{{fig/}}

\usepackage{mathrsfs}
\usepackage{dsfont}
\usepackage{setspace}
\usepackage[export]{adjustbox}
\usepackage{makecell}
\usepackage{enumitem}
\setlist{noitemsep}
\usepackage{verbatim}
\usepackage[normalem]{ulem}

\usepackage{physics}
\usepackage{siunitx}
\usepackage{cleveref}
\crefname{fact}{fact}{facts}
\Crefname{fact}{Fact}{Facts}

\crefname{proposition}{proposition}{proposition}
\Crefname{proposition}{Proposition}{Proposition}

\usepackage{xcolor}
\usepackage{diagbox} %

\usepackage{amsmath,amssymb,amsthm}
\usepackage{mathtools}%
\usepackage{thmtools}
\DeclarePairedDelimiter\ceil{\lceil}{\rceil}

\usepackage{algorithm}
\usepackage{algpseudocodex}
\algrenewcommand\algorithmicrequire{\textbf{Input:}}
\algrenewcommand\algorithmicensure{\textbf{Output:}}

\newtheorem{theorem}{Theorem}
\newtheorem{lemma}[theorem]{Lemma}
\newtheorem{proposition}[theorem]{Proposition}
\newtheorem{corollary}[theorem]{Corollary}
\newtheorem{fact}[theorem]{Fact}
\newtheorem{definition}[theorem]{Definition}

\usepackage{stmaryrd}
\newcommand{\dsl}[0]{\ensuremath{\llbracket}}
\newcommand{\dsr}[0]{\ensuremath{\rrbracket}}
\newcommand{\F}[0]{\ensuremath{\mathbb{F}_{2}}}
\newcommand{\expect}[2][]{\ensuremath{\underset{#1}{\mathbb{E}}\left[#2\right]}}

\newcommand{\splitatcommas}[1]{%
  \begingroup
  \ifnum\mathcode`,="8000
  \else
    \begingroup\lccode`~=`, \lowercase{\endgroup
      \edef~{\mathchar\the\mathcode`, \penalty0 \noexpand\hspace{0pt plus 1em}}%
    }\mathcode`,="8000
  \fi
  #1%
  \endgroup
}

\begin{document}

\title{Fast quantum interconnects via constant-rate entanglement distillation}
\author[1,2,*,$\dagger$]{Christopher A. Pattison}
\author[3,4,*]{Gefen Baranes}
\author[3]{J. Pablo Bonilla Ataides}
\author[3]{Mikhail D. Lukin}
\author[2,3,$\dagger$]{Hengyun Zhou}

\affil[1]{Institute for Quantum Information and Matter, California Institute of Technology, Pasadena, CA 91125, USA.}
\affil[2]{QuEra Computing Inc., 1284 Soldiers Field Road, Boston, MA, 02135, USA.}
\affil[3]{Department of Physics, Harvard University, Cambridge, Massachusetts 02138, USA.}
\affil[4]{Massachusetts Institute of Technology, 77 Massachusetts Avenue, Cambridge, MA 02138, USA.}
\affil[*]{These authors contributed equally.}
\affil[$\dagger$]{Correspondence should be addressed to: cpattiso@caltech.edu, hyzhou@quera.com.}
\twocolumn[
  \begin{@twocolumnfalse}

  \maketitle

    \begin{abstract}
    Distributed quantum computing allows the modular construction of large-scale quantum computers and enables new protocols for blind quantum computation.
    However, such applications in the large-scale, fault-tolerant regime place stringent demands on the fidelity and rate of entanglement generation which are not met by existing methods for quantum interconnects.
    In this work, we develop constant-rate entanglement distillation methods to address this bottleneck in the setting of noisy local operations.
    By using a sequence of two-way entanglement distillation protocols based on quantum error detecting codes with increasing rate, and combining with standard fault tolerance techniques, we achieve constant-rate entanglement distillation.
    We prove the scheme has constant-rate in expectation and further numerically optimize to achieve low practical overhead subject to memory constraints.
    We find our optimized schemes outperform existing computationally efficient quantum interconnect schemes by an order of magnitude in relevant regimes, leading to a direct speed-up in the execution of distributed quantum algorithms.
    \end{abstract}
\vspace*{0.2in}
  \end{@twocolumnfalse}
]

\section{Introduction}
The ability to interconnect multiple quantum nodes is key to a wide range of tasks in quantum communication, distributed quantum metrology, and distributed quantum computing~\cite{kimble2008quantum, Wehner2018QuantumInternet, Azuma2023QuantumRevModPhys}.
By exploiting entanglement between multiple network nodes, one can perform information-theoretic-secure quantum key distribution~\cite{bennett1984update,Ekert1991Bell, gisin2002quantum,lo2014secure}, or realize optical interferometry with very long baselines~\cite{gottesman2012longer,khabiboulline2019optical,khabiboulline2019quantum}.
Furthermore, a high-fidelity quantum interconnect can enable distributed quantum computing between multiple nodes~\cite{buhrman2003distributed, monroe2012large}, as well as blind quantum computation where an untrusted server can perform a desired computation without gaining knowledge about the computation itself~\cite{Broadbent2009BQC, fitzsimons2017private}.
The ability to distribute the computation among multiple nodes may also have additional benefits in terms of modularity and interchangeability of hardware systems, enabling easier scaling to large system sizes.

\begin{figure*}
  \includegraphics[width=\textwidth]{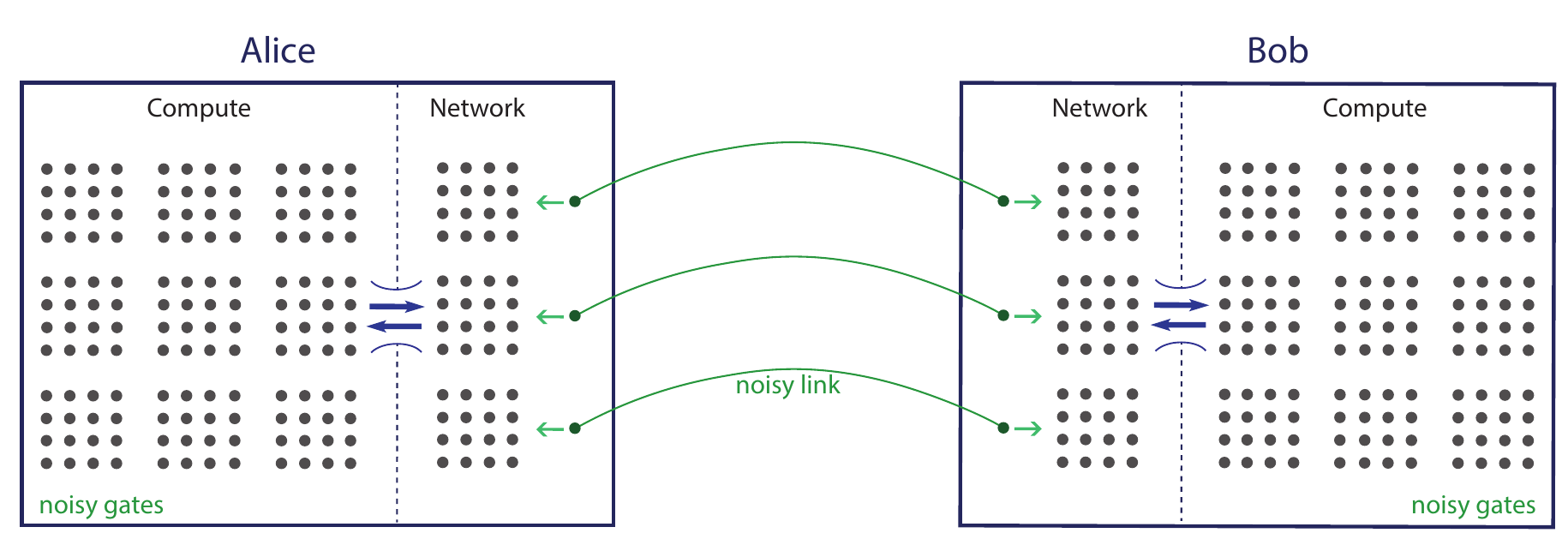}
  \caption{Overview of the distributed quantum computing model with two nodes, Alice and Bob. Each node consists of a ``Compute'' area for local quantum operations and a ``Network'' area, a buffer memory space for managing entangled qubits. Logical qubits in the network are entangled over a noisy channel, then distilled and transferred to the compute areas to facilitate distributed quantum computations.}
  \label{fig:setup}
\end{figure*}

Despite these promising applications, the capabilities of current generation quantum interconnects are still very far from the desired fidelities and communication rates.
This is particularly challenging for applications involving computation, where one has stringent demands on both the achievable entanglement fidelity as well as the resulting rate of quantum communication.
For example, state-of-the-art demonstrations of interconnecting quantum nodes only achieve entangling fidelities in the high 90s and entangling rates of a few hundred pairs per second~\cite{stephenson2020high,jing2019entanglement,ritter2012elementary,hucul2014modular,mirhosseini2020superconducting,knaut2023entanglement,pompili2021realization,meesala2023quantum}.
The entanglement fidelity can be improved through the use of various one-way or two-way entanglement distillation schemes~\cite{bennett1996purification,bennett1996mixed}, but this often leads to a further degradation in the effective logical entangling rate.
With typical target gate operation infidelities for large-scale algorithms in the $10^{-12}$ range, the standard 4-to-1 recursive distillation scheme would require almost a hundred physical Bell pairs (per logical Bell pair) of percent-level infidelity.
Alternative schemes relying on lattice surgery have also been developed~\cite{fowler2010surface,ramette2023fault}, but would again require hundreds to thousands of physical Bell pairs per logical Bell pair in similar parameter regimes.
The combination of these factors suggests that current computationally efficient schemes and physical hardware will result in logical entangling rates of a few Hz, much slower than the QEC cycle times of 1~$\mu$s to 1~ms in various state-of-the-art hardware systems~\cite{acharya2022suppressing,bluvstein2024logical,ryan-anderson2021realization,egan2021fault,abobeih2022fault,postler2023demonstration,krinner2022realizing,huang2023comparing,sivak2022real,ni2022beating}.
It is thus highly desirable to develop methods that use the communication link much more sparingly, thereby achieving higher logical entangling rates while maintaining high target logical fidelities.

In this paper, we develop entanglement distillation methods that achieve constant quantum communication rate in expectation, and show that they yield practical protocols with very low communication overhead.
To achieve the constant communication overhead, we adapt a technique originating in fault-tolerant computation~\cite{yamasaki2022time}, where careful selection of the code concatenation sequence permits a concatenated code family to achieve constant encoding rate.
We rigorously analyze the error rates and success probabilities at each step of the protocol, thereby providing theoretical guarantees of the distillation overhead and performance.

Fig.~\ref{fig:setup} illustrates the setup we consider in this article. We work in the setting where two-way classical communication is permitted~\cite{bennett1996mixed}, thereby allowing the use of error detection and post-selection instead of error correction.
This allows one to use input physical Bell pairs with much lower starting fidelity and can achieve higher encoding rates.
The use of error detection allows high encoding rate without needing to solve the decoding problem for quantum stabilizer codes.
In the computationally unbounded setting, random quantum stabilizer codes (hashing) can also be used to achieve high encoding rate in error-correction mode.
However, decoding of random quantum stabilizer codes is closely related to decoding random \emph{classical} codes, which is used as a hardness assumption for certain post-quantum cryptosystems~\cite{mceliece1978public,bernstein2017classic,melchor2018hamming,aragon2022bike}, and therefore may be computationally challenging.
Constant-rate coding schemes such as quantum CSS polar codes~\cite{renes2012efficient} may be used, but they do not achieve the coding rate of random stabilizer codes on Pauli channels.

To conserve communication bandwidth and ensure that errors primarily arise from noisy physical Bell pair generation between different nodes, we inject each physical qubit state into a logical qubit (e.g. a surface code) using standard high-fidelity state injection techniques~\cite{li2015magic,lao2022magic,gidney2023cleaner}, which add only a negligible amount of noise in the process when the network is much noisier.
This allows us to bound and ignore local operation errors in the process, thereby using the communication channel more efficiently.

Following the proof of asymptotic constant rate, we then examine our scheme numerically, directly optimizing the code sequence over the best possible parameters for small classical and quantum codes~\cite{grassl2023codetables}.
We find sequences of codes that require as few as 6 physical Bell pairs per logical Bell pair at 1\% physical Bell pair error rates, almost an order of magnitude lower than conventional schemes, with only a modest increase in the required memory footprint.
Our scheme also allows a continuous tradeoff between communication overhead and network buffer memory footprint, enabling the optimization of these parameters in a wide range of physical settings.

Our techniques are particularly useful in the regime where each network node is of a reasonably large size, such that there is sufficient buffer memory space to store the logical qubits used for distillation.
Recent hardware advances have demonstrated the possibility of local nodes with hundreds to thousands of physical qubits in a variety of physical platforms~\cite{ebadi2021quantum,ibm2023thousand,malinowski2023how,tao2023high,manetsch2024tweezer,norcia2024iterative}, suggesting that this could be a very relevant regime for distributed quantum computation.
This is in contrast to the setting of small, noisy local nodes typically analyzed in the setting of quantum networks~\cite{muralidharan2016optimal,gisin2002quantum,kimble2008quantum}.
We therefore believe that our techniques will be useful for large scale, fault-tolerant distributed quantum computing.

This paper is organized as follows:
We first provide a high-level summary of our results in the remainder of this section.
In Sec.~\ref{sec:background}, we introduce some basic notation and relevant background definitions.
We then describe our scheme for concatenated entanglement distillation, generalized from Ref.~\cite{bennett1996mixed}, in Sec.~\ref{sec:distillation_setup}.
We show that this scheme achieves asymptotically constant rate in Sec.~\ref{sec:constant_overhead}, and numerically analyze its performance in Sec.~\ref{sec:numerics}, together with its implications on distributed quantum computing in Sec.~\ref{sec:numerics} and Sec.~\ref{sec:distributed_computation}.
Finally, we conclude with a discussion in Sec.~\ref{sec:conclusion}.

\subsection{Summary of results}
Our analysis focuses on the setting of distributing a large quantum computation across multiple networked quantum nodes (Fig.~\ref{fig:setup}), and provides methods to significantly reduce the communication overhead over the network, which is typically the bottleneck in experimental implementations~\cite{stephenson2020high,jing2019entanglement,ritter2012elementary,hucul2014modular,mirhosseini2020superconducting,knaut2023entanglement,pompili2021realization,meesala2023quantum}.

Using the correspondence between entanglement distillation and quantum error correction~\cite{bennett1996mixed}, we adapt recent schemes to achieve constant space overhead computation via code concatenation~\cite{yamasaki2022time} to the quantum interconnect setting (Fig.~\ref{fig:concat-distillation-process}).
Unlike those results, however, we exploit the ability to herald successful attempts in the state distillation setting, thereby using quantum error detecting codes to simplify the distillation procedure and improve performance.

\begin{figure*}
  \centering
  \includegraphics[width=\textwidth]{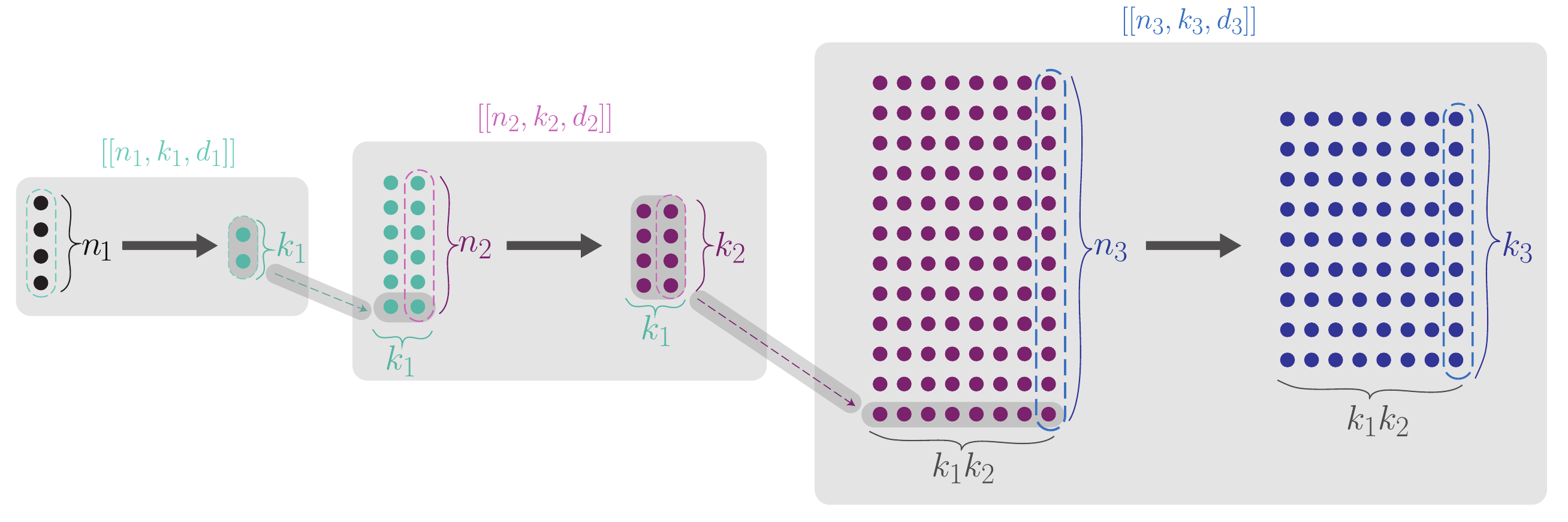}
  \caption{
    Repeated concatenated distillation procedure.
    At each step the checks of a quantum code are imposed along each column.
    If the syndrome is trivial, then each column is unencoded and the entire register becomes a row in the next step.
    Each round of distillation uses qubits originating from independent lower-level distillation attempts to avoid error correlations in the input.
    \label{fig:concat-distillation-process}
  }
  \label{fig:distillation-concatenation}
\end{figure*}

\subsubsection{Constant Rate Entanglement Distillation}
Our main formal result is the construction of an entanglement distillation scheme that achieves a quantum communication rate (distillation overhead) that is asymptotically constant in expectation, meaning that the ratio between output logical Bell pairs and input physical Bell pairs is constant, thereby using the quantum interconnect more efficiently.

\begin{theorem}[Informal statement of main theorem]
\label{thm:informal_thm}
For Bell pairs subjected to a noise channel \(\mathcal{E}\) that acts as the identity with probability \(1-p\) and applies an arbitrary Pauli error with probability \(p \in [0,1/2)\), there exists an entanglement distillation scheme based on concatenated error detection, such that for any desired output error rate \(\epsilon \in (0,1)\), the following holds:
\begin{enumerate}
    \item Let \(K\) be the number of distilled Bell pairs of fidelity \(\epsilon\) output and let \(N\) be the number of Bell pairs consumed, then \footnote{The notation \(f(x) = \Omega(g(x))\) means that \(g(x) = O(f(x))\).}
        \begin{align*}
            \expect{\frac{K}{N}} = \Omega(1)
        \end{align*}
    \item For some \(\alpha >0\), the maximum amount of memory required is
        \begin{align*}
            O\left(\left(\log \log 1/\epsilon\right)^{\alpha \log \log 1/\epsilon}\right),
        \end{align*}
        which is quasi-polylogarithmic in \(1/\epsilon\).
    \end{enumerate}
\end{theorem}

In Thm.~\ref{thm:informal_thm}, we assume perfect local operations.
Using standard state injection techniques for the surface code~\cite{li2015magic,lao2022magic,gidney2023cleaner} (see Sec.~\ref{sec:injection}), we can remove this restriction to achieve the same result with a marginally lower Bell pair fidelity threshold, as long as the local operations have noise rate $p_\mathrm{gate}\lessapprox 0.3\%$ (below the surface code threshold).

To prove this theorem, we will first bound the logical error rate and success rate of the $[[n,n-2,2]]$ parity code (\cref{fact:basic-ed-code}), when operated as a quantum error detecting (QED) code for entanglement distillation.
This code quadratically suppresses the logical error rate (\cref{prop:ed-distillation-error-rate}), with a failure rate that is linear in the input error rate (\cref{prop:ed-overhead}).
By concatenating the protocol using a sequence of such codes with $n_i=(2i)^2$ (\cref{defn:quadratic-code-sequence}), the overhead and success rate at each layer
\begin{align}
\frac{n}{k}=\frac{n}{n-2}\rightarrow 1, \quad
p_{success}=1-O(\sqrt{p_{out}})\rightarrow 1,
\end{align}
leading to the combined constant rate via careful analysis of the failure and encoding rates.
Moreover, because each level of concatenation suppresses the error rate quadratically, the output error rate scales as
\begin{align*}
p_{out}\propto (\mathrm{const.})^{2^l}.
\end{align*}
Achieving a logical error rate of $\epsilon$ therefore only requires $O(\log\log 1/\epsilon)$ levels of concatenation, leading to the bound on the total amount of memory.

\subsubsection{Optimization in Practical Regimes}
After proving that constant quantum communication overhead using error detection is possible, we then turn our attention to an analysis at finite size, where we show that our methods can offer an order-of-magnitude reduction in communication overhead in practically-relevant regimes.
Optimization of the exact code sequence is of practical interest in order to both reduce the amount of memory consumed by the distillation procedure and achieve the best communication overhead.
Our methods allow a continuous trade-off between communication overhead and memory consumption, allowing us to identify the most efficient sequence given target constraints.

Figure~\ref{overheads_different_p_in} illustrates the trade-off between communication overhead and buffer size across various error rates. The data points represent optimized distillation sequences, showing that our scheme maintains high communication efficiency even with modest buffer sizes of tens of logical qubits. The overhead decreases consistently across different error rates as the buffer size increases, allowing a continuous trade-off between different parameters.

In Tab.~\ref{table:comparison}, we further 
present a comparison of the communication overhead for different schemes under various network error rates. Our constant-rate distillation scheme significantly reduces resource overhead compared to existing computationally efficient schemes~\cite{bennett1996mixed,fowler2010surface,ramette2023fault}. For instance, with a $1\%$ Bell pair error rate, our scheme achieves an overhead of around 6 physical Bell pairs per logical Bell pair, while BDSW-2EPP~\cite{bennett1996mixed} requires dozens and lattice surgery requires thousands. This highlights the practical advantages of our scheme in reducing resource demands.

Finally, we consider the detailed scheduling of the full distillation process, including a pipelined approach to minimize space-time overhead, and compare the results to other methods. This comprehensive analysis confirms the practical applicability and efficiency of our scheme for distributed quantum computation.

\subsubsection{Distributed Quantum Computation}
With the promising performance of our entanglement distillation sequences, we analyze the implications of these results in the setting of distributed quantum computation.
We find that a useful characterization of the communication requirements is the algorithm- and compilation-dependent parameter $\beta$, defined to be the average number of intercore logical entanglement operations per core required for each intracore logical circuit layer.
This parameter may typically range from $O(1)$ for sequential circuits, to $O(s_c)$ for random circuits, where $s_c$ is the number of computational logical qubits per core.
The network communication becomes the bottleneck when $\beta t_{inter}\geq t_{intra}$, where $t_{inter}$ is the inter-module logical operation time and $t_{intra}$ is the intra-module logical operation time.
Since most current quantum computing platforms have much faster intra-module operation speeds compared to inter-module communication~\cite{stephenson2020high,jing2019entanglement,ritter2012elementary,hucul2014modular,mirhosseini2020superconducting,knaut2023entanglement,pompili2021realization,meesala2023quantum}, one is typically in the limit of network communication being the bottleneck.
We therefore expect that the order-of-magnitude reduction in communication overhead we achieve directly leads to the same improvement in executing large-scale, distributed quantum computations.
Our methods are also directly compatible with recent methods to speed up fault-tolerant computation on local nodes, based on transversal gates and algorithmic fault tolerance~\cite{cain2024correlated,zhou2024algorithmic}.

\begin{figure}
  \centering
  \includegraphics[width=\columnwidth]{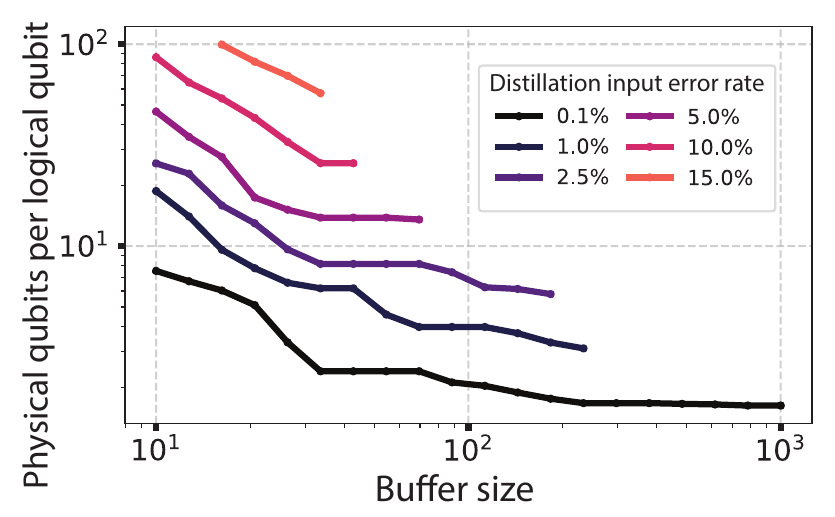}
  \caption{Optimized physical Bell pair overhead per logical Bell pair, as a function of the buffer size constraint, for different logical Bell pair input error rates. Each data point represents a distillation sequence, with a target output error rate $p_{\text{target}} = 10^{-12}$.
  \label{overheads_different_p_in}
  }
\end{figure}

\section{Background}
\label{sec:background}
\subsection{Task of Quantum Interconnects}

We consider two quantum network nodes with noisy local gates, denoted Alice and Bob, connected by a noisy depolarizing channel \(\mathcal{E}\).
The task of a quantum interconnect is to prepare a \(2N\)-qubit entangled state between Alice and Bob with fidelity \(\epsilon\) to the \(N\) Bell pair state \(\ket{\Phi_+}^{\otimes N}\), with as few uses of the noisy quantum channel as possible.
Using local operations and classical communication (LOCC), many tasks can then be carried out between the two parties using Bell pairs as a resource state, including high-fidelity logical qubit teleportation and distributed quantum computation across multiple nodes (Sec.~\ref{sec:distributed_computation}).

We reduce the problem to the setting of a noisy quantum channel and perfect LOCC, by using state injection to encode each qubit into a quantum code block of sufficient size, such that logical operation errors are negligible compared to the target Bell pair fidelity $\epsilon$.
Due to this reduction, the theory results of \cref{sec:distillation_setup} and \cref{sec:constant_overhead} are proved in the setting where both parties are capable of perfect LOCC.
We will quantitatively analyze the overhead incurred by the additional encoding in Sec.~\ref{sec:injection}, where we compare the practical overheads required for different quantum interconnect schemes.

\subsection{Definitions and Notation}
We now introduce the notation that we will use for the rest of the paper, and review some related concepts.
Let \([n]\) denote the set of natural numbers from \(1\) to \(n\).
We denote the \(n\)-qubit Pauli group \(\{I,X,Y,Z\}^{\otimes n}\) by \(\mathcal{P}_{n}\).
We denote the simultaneous \(+1\) eigenstate of \(Z\otimes Z\) and \(X\otimes X\) by \(\ket{\Phi_+}\).
For an indexed set of operators \(\mathcal{O} = \{O_i\}_{i\in [n]}\) and a vector \(a \in \F^{n}\), we use the abuse of notation \((-1)^{a}\mathcal{O}\) to denote the indexed set \(\{(-1)^{a[i]} O_i\}_{i\in [n]}\).\footnote{After associating \(\F\) with \(\{0,1\} \subseteq \mathbb{Z}\) in the canonical way.}
For an operator \(U\), we also use the abuse of notation \(U\mathcal{O} = \{UO_i\}_{i\in [n]}\).
In the context of two subsystems \(A\), \(B\), for an operator \(P\), where necessary we will use \(P_{i}^{(A)}\) to denote \(P\) supported on the \(i\)-th qubit of the subsystem \(A\) and likewise for subsystem \(B\).

We will focus on Pauli noise channels for our noisy communication link, where using the properties of the Bell pair, we assume that the noise only acts on one qubit of the Bell pair.
On occasion, we will need to refer to the corresponding distribution over Pauli operators instead of the channel directly.
\begin{definition}[Pauli noise]
  For \((p_x,p_y,p_z) \in \{(x,y,z) \in \mathbb{R}_+^3\mid x+y+z \le 1\}\), we define the \((p_x,p_y,p_z)\)-Pauli noise channel as the channel \(\rho \mapsto (1-p_x-p_y-p_z) \rho + p_x X \rho X + p_y Y \rho Y + p_z Z \rho Z\).
  The \(p\)-depolarizing channel, \(p \in [0,3/4]\), is the \((p/3, p/3, p/3)\)-Pauli noise channel.
  Finally, we define a \(p\)-Pauli noise channel, \(p < 1\), as any \((p_x,p_y,p_z)\)-Pauli noise channel such that \(p_x+p_y+p_z \le p\).
\end{definition}

We use the stabilizer formalism~\cite{gottesman1997stabilizer} to describe the quantum codes that we work with, and use the standard notation \(\dsl n, k, d\dsr\) to denote a quantum code with $n$ data qubits, $k$ encoded logical qubits, and code distance $d$.
A stabilizer code is described by its stabilizer group $S$, which is an Abelian subgroup of the $n$-qubit Pauli group with $-1\notin S$.
The codespace is then given by the set of states $|\psi\rangle$ that satisfy $s|\psi\rangle=|\psi\rangle$ for all $s\in S$.
The encoding map, \(\mathcal{D}^{-1}\), for an \(\dsl n, k, d \dsr\) stabilizer code takes as input \(k\) bare qubits and outputs the \(n\)-qubit code state with the input \(k\) qubits as logical qubits.
The unencoding map, \(\mathcal{D}\), takes as input the \(n\)-qubit code state and outputs the stored \(k\) logical qubits as bare qubits.

The stabilizer formalism also allows us to easily track errors in the circuit.
Let \(S=\{s_i\}_{i\in [r]}\) be \(r\) generators of the stabilizer code, and let \(E\) be a Pauli error.
The corresponding syndrome is then given by the string \(a \in \F^r\) such that \(E s_i E = (-1)^{a_i}S\).
By using the syndrome information, we can either apply a recovery operation to correct the observed error, up to stabilizer equivalences (quantum error correction, QEC), or discard the results if a nontrivial syndrome is observed (quantum error detection, QED).
The former only requires one-way communication, since we only need to send the observed syndromes from one party to the other; the latter requires two-way communication so that both parties know which qubits to discard, but it can handle higher input error rates~\cite{bennett1996mixed}.
Note that one can also apply the quantum code in a hybrid mode, in which we perform error detection on a certain set of syndromes but perform error correction on other syndromes.

More generally, we can use the stabilizer formalism to describe and track the state of a quantum system at any given point.
Consider a stabilizer state on $n$ qubits $|\psi_S\rangle$, which is the unique state that satisfies $s|\psi\rangle=|\psi\rangle$ for all $s\in S$, with $S$ being an Abelian group with $n$ independent generators.
When a unitary $U$ is applied to the state $|\psi\rangle$, we can update the stabilizer generators as $s\rightarrow UsU^\dagger$.
When we measure a product of Pauli operators $\mathcal{M}$ and obtain the outcome $a_{\mathcal{M}}$, there are two possibilities.
If $\mathcal{M}$ commutes with all elements of $S$, then the measurement result will be deterministic, and the stabilizer group remains unchanged.
Otherwise, we update the stabilizer group as follows: let the set of generators that anti-commute with $\mathcal{M}$ be $s_i$, $i\in \{1,...,r\}$, then the new stabilizer group is generated by the elements
$(-1)^{a_{\mathcal{M}}}\mathcal{M}$, $s_1s_i$ for $i\in \{2,...,r\}$ and $s_i$ for $i\in \{r+1,...,n\}$.

\section{Entanglement Distillation}
\label{sec:distillation_setup}
We can use a quantum error-detecting (QED) or a quantum error-correcting (QEC) code in order to perform distillation.
We begin by defining a basic protocol using two-way communication and QED codes that refines Bell pairs between Alice and Bob. 
For an alternate presentation on the use of stabilizer codes for entanglement distillation, see \cite{matsumoto2003conversion,wilde2010convolutional,shi2024stabilizer}.

\begin{definition}[Error-detection distillation \cite{bennett1996mixed}]\label{defn:ed-distillation}
  Let \(S_{\mathrm{ED}} = \{s_i\}_{i\in [r]}\) be a set of \(r\) independent stabilizer generators for an \(\dsl n, k, d \dsr\) stabilizer code with encoding map \(\mathcal{D}^{-1}\).
  Let \(\mathcal{E}\) be a Pauli error channel and \(\mathcal{I}\) be the identity channel.
  Define the state between Alice and Bob \(\rho_{AB} = (\mathcal{I}\otimes \mathcal{E})\left(\ket{\Phi_+}_{AB}\right)\) given by subjecting one half of a Bell pair to the error channel.
  We can define the following distillation protocol:
  At the start of the protocol, Alice and Bob jointly possess the state \(\rho^{\otimes n}_{AB}\).
  \begin{algorithmic}[1]
    \Require{\(n\) noisy Bell pairs \(\rho_{AB}^{\otimes n}\)}
    \Ensure{\(k\) less noisy Bell pairs \(\rho'_{AB}\) or FAIL}
    \State Alice measures \(S_{\mathrm{ED}}\) on their qubits, producing measurement outcomes \(a \in \F^r\).
    \State Alice computes a correction \(R\) that corrects the syndrome \(a\).
    \State Bob measures \(S_{\mathrm{ED}}\) on their qubits, producing measurement outcomes \(b \in \F^r\).\label{alg:distill:measurements}
    \State Alice and Bob exchange \(a\), \(R\), and \(b\).\label{alg:distill:communication}
    \If{\(\sigma=a+b\) is non-zero}
      \State \Output FAIL
    \Else
      \State Apply \(R \otimes R\) to the joint state.\label{alg:distill:recovery}
      \State \Output \(\mathcal{D}\otimes \mathcal{D}\) applied to the joint state.\label{alg:distill:unencode}
    \EndIf
  \end{algorithmic}
\end{definition}

Since this scheme may output FAIL, we refer to it as probabilistic.
Likewise, we refer to a scheme that always produces an output state and never produces FAIL as deterministic.

Line~\ref{alg:distill:recovery} is not necessary, but we include it so that the state before the application of \(\mathcal{D}\) is in the codespace.
We note that there is another, deterministic, variant of this protocol based on \emph{error correction} where Bob always applies a correction according to \(\sigma\) and the \(k\) Bell pairs are always output.

An important feature of the distillation scheme of \cref{defn:ed-distillation} is the ability to detect or correct errors during transmission, despite the transmitted states initially being unencoded.
The key intuition is that due to the physical Bell pair entanglement, projecting one side into an eigenspace of the stabilizer generators will automatically do the same on the other side, thereby guaranteeing the two sides to have the same measurement result in the absence of errors.
Thus, it is the difference between Alice and Bob's measurements that detects errors and serves as the syndrome.
We illustrate this in \cref{fig:qed-distillation-stabilizers} and formalize this in the following statements, with proofs deferred to \cref{app:distillation-proofs}.

\begin{figure*}
\begin{center}
  \begin{quantikz}[every matrix/.append style={name=mycd}, execute at end picture={
      \draw ($(mycd-1-1)!0.35!(mycd-1-2)$) node(L1){}
            ($(mycd-1-2)!0.43!(mycd-1-3)$) node(L2){}
            ($(mycd-1-3)!0.53!(mycd-1-4)$) node(L3){}
            ($(mycd-1-4)!0.53!(mycd-1-5)$) node(L4){}
            ($(mycd-1-5)!0.65!(mycd-1-6)$) node(L5){};
       \draw[thick] (L2) -- ++(-0.8, 0.5) node[anchor=south east] {\(E\{X_i\otimes X_i, Z_i\otimes Z_i\}_{i \in [n]}E\)};
       \draw[thick] (L3) +(0,0.2) -- ++(-0.1, 1.2) node[anchor=south] {\(E \mathcal{L}_{\mathrm{Bell}} E \sqcup (-1)^{a} S_{\mathrm{ED},A} \sqcup (-1)^{a+\sigma}S_{\mathrm{ED},B}\)};
       \draw[thick] (L4) -- ++(0.4, 0.4) node[anchor=south west] {\(E \mathcal{L}_{\mathrm{Bell}} E \sqcup S_{\mathrm{ED},A} \sqcup (-1)^{\sigma}S_{\mathrm{ED},B}\)};
  }]
    \setwiretype{n} & \slice{} & \slice{} & \slice{} \setwiretype{c} \rstick{\(a\)}       &  \setwiretype{n}& \\
    \makeebit{}     & \qwbundle{n}       & \gate{\mathcal{M}_{\mathcal{S}}}\wire[u]{c}   & \gate{\mathcal{R}_{a}} & \gate{\mathcal{D}} & \qwbundle{k} \\
                    & \gate{\mathcal{E}} & \gate{\mathcal{M}_{\mathcal{S}}}\wire[d]{c}   & \gate{\mathcal{R}_{a}} & \gate{\mathcal{D}} & \qwbundle{k} \\
    \setwiretype{n} &                    &                                               & \setwiretype{c}  \rstick{\(a+\sigma\)}       & \setwiretype{n}&
  \end{quantikz}
  \caption{Stabilizer generators of the QED distillation protocol with an error \(E\in \mathcal{P}^{\otimes n}\) applied by the channel \(\mathcal{I}_A \otimes \mathcal{E}_B\), at different locations of the circuit: 1) after initialization 2) after measuring stabilizer generators 3) after recovery to the codespace.
    Alice and Bob initially jointly possess \(n\) Bell pairs.
    They then measure the checks of a quantum code, projecting the state into one of \(2^{n-k}\) subspaces.
    By exchanging measurements, they can recover to the codespace and detect an error \(E\) applied by the channel \(\mathcal{E}\) if \(E\) is detectable by the quantum code.
    The figure is labeled by the generators of the stabilizer group of the joint state (defined in the proof of \cref{prop:correctness} in \cref{app:distillation-proofs}).
    \(S_\mathrm{ED}\) are the stabilizer generators of the entanglement distillation code with \(A\) or \(B\) labeling the subsystem, and \(\mathcal{L}_\mathrm{Bell}\) are the logical operators of the entanglement distillation code in a Bell basis.
    Two-way communication is used when determining whether to abort the current distillation attempt.
    \label{fig:qed-distillation-stabilizers}}
\end{center}
\end{figure*}
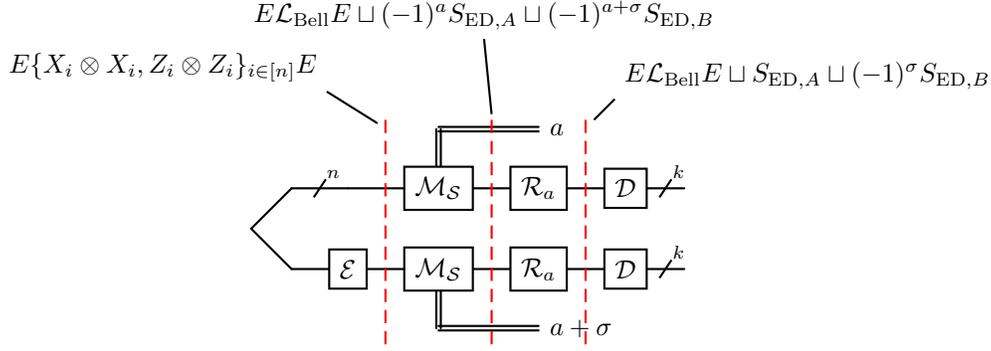

\begin{restatable}[Correctness]{proposition}{propDistillCorrectness}\label{prop:correctness}
  When \(\mathcal{E} = \mathcal{I}\), the distillation scheme of \cref{defn:ed-distillation} outputs \(\ket{\Phi_{+}}^{\otimes k}\).
\end{restatable}

\begin{restatable}[Error correction]{proposition}{propDistillErrorCorr}\label{prop:ed-distillation-ec}
  Consider the distillation scheme of \cref{defn:ed-distillation}.
  Let \(E^{(B)}\) be a Pauli noise operator applied by the noise channel to the \(n\) qubits of \(B\). Then
  \(\sigma = a+b\) is the syndrome of \(E^{(B)}\), and the post recovery state has the error \(E^{(B)}\) on the \(B\) subsystem.
\end{restatable}
Note that here, the post recovery state refers to the state after applying $R$, but not any further error correction based on the distillation code employed.

Let us now introduce the concrete instance of QED code we will use for our proof of constant overhead.
\begin{fact}[Quantum parity code~\cite{steane1996simple,gottesman1998theory,self2022protecting}] \label{fact:basic-ed-code}
For even integer \(n\), the operators \(\{\prod_{i \in [n]}X_i, \prod_{i \in [n]}Z_i\}\) define a CSS code with parameters \(\dsl n, n-2, 2\dsr\).
We refer to this code family as the quantum parity code, in analogy to the classical parity code with one parity check.\footnote{It also is known as the ``iceberg'' code due to the appearance of its Tanner graph when drawn with a blue color scheme.}
\end{fact}

We next establish bounds on the output error rate and the rejection probability.
While we will need to prove generalized versions of these results in the next section, they serve to illustrate the heart of the construction in a simple setting.

\begin{restatable}[Output error rate]{proposition}{propOutputErrorRate}\label{prop:ed-distillation-error-rate}
  Consider the execution of the protocol of \cref{defn:ed-distillation}, where \(\mathcal{E}\) is \(p\)-Pauli noise and the code is the \(\dsl n, n-2, 2\dsr\) code from \cref{fact:basic-ed-code} for \(p < 1/(2n)\) and \(n \ge 4\).
  Conditioned on the event that the protocol did not fail (\(\lnot FAIL\)), the probability that there is no error on the output qubits, \(1-p_{out}\) is at least
  \begin{align}
    1-p_{out} \ge 1-\left(n \frac{p}{1-p}\right)^2
  \end{align}
\end{restatable} 
The main feature of the output error rate is the quadratic scaling in the input noise strength, as the distance-2 QED code can detect any single Pauli error. The full expression is proven in \cref{app:distillation-proofs} by bounding the probability of having at least 2 errors when the protocol does not output fail.

We would like to use the entanglement distillation protocols in settings where a state must always be output.
Therefore, we will ``wrap'' the protocol such that we retry any time FAIL is output.
\begin{definition}[Repeat-until-success]\label{defn:repeat-until-success}
  Given a probabilistic protocol \(P\) and a protocol \(Q\) for producing the inputs to \(P\), we can turn \(P\) into a deterministic protocol requiring no inputs in the following way:
  \begin{algorithmic}[1]
    \Repeat
    \State \(a \gets \) output of \(Q\)
    \State \(b \gets \) output of \(P(a)\)
    \Until{\(b\) is not FAIL}
    \State \Output \(b\)
  \end{algorithmic}
  We refer to the number of times \(Q\) is executed as the attempt count.
\end{definition}

\begin{restatable}[Attempt count]{lemma}{lemmaAttemptCount}\label{prop:retry-count}
  Let \(p_{\mathrm{fail}}\) be the failure rate of the probabilistic protocol in \cref{defn:repeat-until-success}.
  Then, the average attempt count is \(\frac{1}{1-p_{\mathrm{fail}}}\) and the attempt count exceeds \(\frac{\log 1/\epsilon}{\log 1/p_{\mathrm{fail}}}\) with probability at most \(\epsilon\).
\end{restatable}

\begin{proof}
  Let \(T\) be the attempt count.
  Before attempt \(t\in \{1,2,\dots\}\), the probability that no prior attempt has been successful is \(p_{\mathrm{fail}}^{t-1}\), so \(\Pr(T \ge t) = p_{\mathrm{fail}}^{t-1}\).
  The probability mass function (PMF) of \(T\) is then given by the discrete difference \(\Pr(T = t) = p_{\mathrm{fail}}^{t-1} - p_{\mathrm{fail}}^{t}\), which we can use to evaluate \(\expect{T}\).
  \begin{align}
    \expect{T} = \sum_{t=1}^{\infty} t \left(p_{\mathrm{fail}}^{t-1} - p_{\mathrm{fail}}^{t}\right) = \frac{1}{1-p_{\mathrm{fail}}}
  \end{align}

  To achieve the tail bound, pick \(t = \ceil{\frac{\log 1/\epsilon}{\log 1/p_{\mathrm{fail}}}}\), so that
  \begin{align}
    \Pr(T \ge \frac{\log 1/\epsilon}{\log 1/p_{\mathrm{fail}}}) &\le p_{\mathrm{fail}}^{\ceil{\frac{\log 1/\epsilon}{\log 1/p_{\mathrm{fail}}}}}\\
&\le \exp\left(\frac{\log \epsilon}{\log p_{fail}} \log p_{fail}\right) \\
                 &= \epsilon
  \end{align}
  where we used \(p_{\mathrm{fail}}^{\ceil{\frac{\log 1/\epsilon}{\log 1/p_{\mathrm{fail}}}}} \leq p_{\mathrm{fail}}^{\frac{\log 1/\epsilon}{\log 1/p_{\mathrm{fail}}}}\).
\end{proof}

We can combine these results to prove an upper bound on the overhead of a single round of error distillation.
\begin{restatable}[Error detection overhead]{proposition}{propErrorDetectionOverhead}\label{prop:ed-overhead}
  Consider the execution of repeat-until-success error detection distillation, where \(\mathcal{E}\) is \(p\)-Pauli noise and the code is the \(\dsl n, n-2, 2\dsr\) code from \cref{fact:basic-ed-code}.
  For \(n \ge 4\) and \(p \le 1/n\), the average number of consumed Bell pairs \(N\) to output \(k\) Bell pairs is
  \begin{align}
    \expect{\frac{N}{k}} \le \left[3 n p + 1\right] \frac{n}{n-2}
  \end{align}
\end{restatable}
Here, the former factor comes from a bound on the expected retry overhead in \cref{prop:retry-count}, while the latter corresponds to the encoding rate of the code. The detailed proof is again deferred to \cref{app:distillation-proofs}.

\section{Constant Overhead}
\label{sec:constant-overhead}
In this section, we generalize the preceding analysis, and show how a concatenated version of the above method allows us to achieve constant overhead entanglement distillation.

To start, we review the two procedures consider in Ref.~\cite{bennett1996mixed}, similar to \cref{defn:ed-distillation}.
The first (BDSW-2EPP) uses two-way communication and the \([2,1,2]\) classical repetition code alternately in the \(X\) and \(Z\) basis, resulting in a rate going to zero.
To achieve finite rate, this can be combined with a second scheme (BDSW-1EPP), which uses one-way communication and a random quantum code as an error correcting code.
However, this has the downside that a random code needs to be decoded, and efficient algorithms for this task are not known.

One approach is to continue using 1EPP as the first stage, and then consider efficiently-decodable code families such as quantum Reed-Muller codes~\cite{steane1996error}.
Note that fault-tolerance is handled by the local logical encoding, so the choice of quantum codes is extremely broad.
One might expect that the rate is suboptimal, because this approach does not utilize the fact that two-way communication is available.

To make progress on this question and to avoid the need to decode, we consider the two-way error-detecting protocol exclusively, but expanding to other codes besides the \([2,1,2]\) classical code.
In particular, there exists constant rate code families based on concatenation~\cite{yamasaki2022time}, which suggests the existence of a concatenation sequence for an error-detecting distillation scheme that results in constant overhead.

We begin by defining a concatenated distillation protocol, which allows us to use a sequence of error detecting codes for distillation.
This becomes somewhat non-trivial when considering codes encoding more than one logical qubit, due to correlations of errors on different logical qubits of the output.
This correlation can be avoided by applying further rounds of distillation to qubits originating from independent lower-level distillation attempts, so that the qubits in a given distillation instance have independent errors.

\label{sec:constant_overhead}
\begin{definition}[Distillation concatenation]\label{defn:distillation-concat}
  Given a distillation protocol \(P_{\mathrm{inner}}\) producing \(k_{\mathrm{inner}}\) output qubits, and a second distillation protocol \(P_{\mathrm{outer}}\) requiring \(n_{\mathrm{outer}}\) input qubits and producing \(k_{\mathrm{outer}}\) output qubits, we can concatenate them to arrive at a protocol producing \(K = k_{\mathrm{inner}} k_{\mathrm{outer}}\) qubits.
  \begin{algorithmic}[1]
    \For{\(i \in [n_{outer}]\)}
      \Repeat
        \State \(t \gets \) output of \(P_{\mathrm{inner}}\) \label{defn:distillation-concat:inner-distillation} \Comment{Inner distillation step}
      \Until{\(t\) is not FAIL}
      \State \(a[i,\cdot] \gets t\) \label{defn:distillation-concat:input-reg}
    \EndFor
    \For{\(j \in [k_{inner}]\)}
      \State \(s \gets\) output of \(P_{\mathrm{outer}}(a[\cdot, j])\) \label{defn:distillation-concat:outer-distillation} \Comment{Outer distillation step}
      \If{\(s\) is FAIL}
        \State \Output FAIL
      \Else
        \State \(b[\cdot,j] \gets s\)
      \EndIf
    \EndFor
    \State \Output b    
  \end{algorithmic}
\end{definition}

\begin{definition}[Concatenated error-detecting distillation]\label{defn:concat-ed-distillation}
  Given a sequence \(\{\mathcal{C}_i\}_{i\in\mathbb{N}}\) of QED codes, we define the \(\ell\)-concatenated QED distillation protocol recursively as follows:
  \begin{itemize}
    \item The \(0\)-concatenated QED distillation simply outputs Bell pairs obtained deterministically\footnote{The underlying process may be probabilistic, but for simplicity we ignore this fact.} from some source without any distillation.
    \item For \(\ell > 1\), the \(\ell\)-concatenated QED distillation protocol is the distillation concatenation protocol using \((\ell-1)\)-concatenated QED distillation as the inner protocol and QED distillation with QED code \(\mathcal{C}_{\ell}\) as the outer protocol.
  \end{itemize}
\end{definition}
This process is illustrated in Fig. \ref{fig:distillation-concatenation} with Alice and Bob jointly executing steps of the protocol.

The choice of QEC code sequence is somewhat arbitrary, but there exists sequences such that the overhead approaches a constant.
One such sequence we call the quadratic parity code sequence.
\begin{definition}[Quadratic parity code sequence]\label{defn:quadratic-code-sequence}
  The quadratic parity code sequence is the code sequence \(\{\mathcal{C}_i\}_{i\in\mathbb{N}}\), where \(\mathcal{C}_i\) is the length \((2i)^2\) quantum parity code from \cref{fact:basic-ed-code}. 
\end{definition}

For later comparison, we also define the sequence of codes used for the BDSW-2EPP scheme of \cite{bennett1996mixed}.
For \(p\)-depolarizing noise with \(p < 1/2\), this sequence results in a decreasing error rate which can be straightforwardly shown by an exact calculation.
\begin{definition}[BDSW-2EPP code sequence \cite{bennett1996mixed}]\label{defn:2epp-code-sequence}
  The BDSW-2EPP code sequence is the sequence \(\{\mathcal{C}_i\}_{i\in\mathbb{N}}\) of \(\dsl 2, 1, 1\dsr\) quantum codes (equiv. \([2,1,2]\) classical repetition codes) where \(\mathcal{C}_i\) is defined by the check \(X_{1}X_{2}\) if \(i\) is odd and \(Z_{1}Z_{2}\) if \(i\) is even.
\end{definition}

We are now ready to state and prove our main result: A constant communication-overhead distillation scheme based on \emph{error detection}.
We prove the asymptotic bound for a given constant input error, but the code sequence can be modified to achieve an identical result for any \(p \in [0,1/2)\) by adding a constant number of rounds of BDSW-2EPP distillation.
\begin{theorem}[Constant communication-overhead distillation]
\label{thm:constant-overhead}
  For physical Bell pairs subjected to a noise channel \(\mathcal{E}\) that is \(p\)-Pauli noise with \(p\in [0,1/2000]\), the \(\ell\)-concatenated error detecting distillation protocol with the quadratic parity code sequence (\cref{defn:quadratic-code-sequence}) satisfies the following properties
  \begin{enumerate}
  \item Conditioned on an output state being produced, the output contains no error with probability \(1-p_\ell\), where
    \begin{align}
      p_\ell \le \frac{1}{34} \left(544 p\right)^{2^\ell}.
    \end{align}
    In other words, to achieve an output block error probability \(\epsilon\), it suffices to pick \(\ell = \ceil{\log_{544 p}\log_2 \left(34 \epsilon\right)} = O(\log \log 1/\epsilon)\).
  \item The protocol consumes \(N\) input Bell pairs to produce \(K\) output Bell pairs (repeat-until-success), where
  \begin{align}
    \expect{\frac{K}{N}} \ge \frac{1}{3}~.
  \end{align}
\item For a desired block output error rate \(\epsilon\), the maximum memory required is \(O\left(\left(\log \log 1/\epsilon\right)^{\alpha \log \log 1/\epsilon}\right)\) for some \(\alpha >0\), which is quasi-polylogarithmic in \(1/\epsilon\).
  \end{enumerate}
\end{theorem}
\begin{proof}[Proof of property 1]
  We proceed by induction in \(\ell \ge 1\) with the additional hypothesis that for \(i \in [\ell]\), \(p_{i-1} \le 1/(2n_i) = 1/8i^2\) defining \(p_0 \equiv p\).
  Let \(a\) be the output of level-\((\ell-1)\) as in line \ref{defn:distillation-concat:input-reg} of \cref{defn:distillation-concat}.
  At level-\(\ell\), the error detection is applied to each column of \(a\), so the output may contain an error only if at least two rows of the input \(a\) contain an error.
  Since each row contains no error with probability \(1-p_{\ell-1}\), an identical argument (requiring the additional inductive hypothesis) and result as \cref{prop:ed-distillation-error-rate} applies with the weight of \(E\) replaced by the number of rows in which \(E\) is non-trivially supported.
  \begin{align*}
    p_\ell &\le \left(n_\ell \frac{p_{\ell-1}}{1-p_{\ell-1}}\right)^2 \\
           &\le n_\ell^2 p_{\ell-1}^2 \left(\frac{1}{1-\frac{1}{2n_{\ell}}}\right)^{-2} \\
           &\le \frac{64}{49} n_\ell^2p_{\ell-1}^2 = \frac{64}{49} 16 \ell^4 p_{\ell-1}^2 \\
           &\le 34 \ell^{2^2} p_{\ell-1}^2
  \end{align*}
   Using the inductive hypothesis, we can evaluate this recurrence relation in closed form (\cref{lemma:error-recurrence,lemma:error-recurrence-product})
  \begin{align}
   p_\ell  &\le \left(\prod_{i=1}^\ell i^{2^{\ell+2-i}}\right)34^{2^\ell-1} p^{2^\ell} \nonumber\\
        &\le 2^{2^{\ell+2}} 34^{2^\ell-1} p^{2^\ell} = \frac{1}{34}\cdot 16^{2^{\ell}} 34^{2^\ell} p^{2^\ell}\nonumber\\
        &\le \frac{1}{34} \left(544 p\right)^{2^\ell} \label{prop:concat-error-rate:inductive-bound}
  \end{align}
  For \(i \ge 1\) and \(p\le 1/2000\), we have that
  \begin{align*}
    \frac{\frac{1}{34}\left(\frac{544}{2000}\right)^{2^i}}{\frac{1}{8(i+1)^2}} \le 1
  \end{align*}
  which establishes \(p_{\ell} \le \frac{1}{8(\ell+1)^2}\).
  For the base case, \(p_0 \le 1/2000 \le 1/8\) and a column contains an error with probability \(p_0\equiv p\), completing the inductive argument.
\end{proof}

\begin{proof}[Proof of property 2]
  We proceed by a similar argument to \cref{prop:ed-overhead}.
  Fix an input error rate \(p \le 1/2000\).
  Let \(p_{\mathrm{fail}}^{(i)}\) be the probability that \(i\)-concatenated QED distillation outputs FAIL, and let \(\mathcal{T}^{(i)}_\mathrm{FAIL}\) be the distribution of the number of times the inner distillation step in \cref{defn:distillation-concat:inner-distillation} is executed by \(i\)-concatenated QED distillation and outputs FAIL.
  \(\mathcal{T}^{(i)}_\mathrm{FAIL}\) is a negative binomial distribution with parameters \(\mathrm{NB}(n_i, 1-p_{\mathrm{fail}}^{(i-1)})\).
  In other words, the distribution is the number of failures in a sequence of trials with success probability \(1-p_{\mathrm{fail}}^{(i-1)}\) until \(n_{i}\) successes are observed.
  
  Let \(\mathcal{X}^{(i)}\) be the distribution of the number of Bell pairs consumed by \(i\)-concatenated QED distillation, defining \(\mathcal{X}^{(0)}\) to be the constant distribution supported on \(1\), \(\mathrm{Bernoulli}(1)\).
  For independent random variables \(X^{(i)} \sim \mathcal{X}^{(i)}\), \(T^{(i)} \sim \mathcal{T}^{(i)}_\mathrm{FAIL}\), and \(\{Y_{j}^{(i-1)}\}_{j\in[n_i]}\) with \(Y^{(i-1)}_{j} \sim \mathcal{X}^{(i-1)}\) i.i.d., we have that
  \begin{align}
    X^{(i)} = \sum_{j=1}^{n_i+T^{(i)}} Y^{(i-1)}_j
  \end{align}
  That is, the number of Bell pairs required for an attempt at \(i\)-concatenated QED distillation is the number of Bell pairs required for \(T^{(i)}+n_i\) attempts at \((i-1)\)-concatenated QED distillation.
  We have independence of the random variables, because the number of Bell pairs required to produce an output of \((i-1)\)-concatenated QED distillation does not affect the probability to produce an output of \(i\)-concatenated QED distillation.

  To prove the bound on expectation, first note that using properties of the expectation
  \begin{align*}
    \expect{X^{(i)}} &= \expect[T^{(i)}]{\expect[\{Y^{(i-1)}_j\}_j]{\sum_{j=1}^{n_i+T^{(i)}} Y^{(i-1)}_j}} \\
&= \expect{n_i+T^{(i)}}\expect{Y^{(i-1)}_1} \\
&= \frac{n_i}{1-p_{\mathrm{fail}}^{(i-1)}}\expect{X^{(i-1)}}
  \end{align*}
  Defining \(p_{\mathrm{fail}}^{(0)} = 0\) and applying this property recursively, we have that an attempt at \(\ell\)-concatenated QED distillation requires \(\prod_{i=1}^{\ell} \frac{n_i}{1-p_{\mathrm{fail}}^{(i-1)}}\) Bell pairs on average, and we must make \(\frac{1}{1-p_{\mathrm{fail}}^{(\ell)}}\) attempts on average to obtain an output.
  As in \cref{prop:ed-overhead}, we will use the bound \(1-p_{\mathrm{fail}}^{(i)} \ge (1-p_i)^{n_i}\) where \(p_i\) is the probability that the output block contains an error.
  The expected number of input Bell pairs \(N\) before \(K\) Bell pairs are output is thus
  \begin{align*}
    \expect{\frac{N}{K}} &= \frac{1}{K}\frac{1}{1-p_{\mathrm{fail}}^{(\ell)}} \prod_{i=1}^{\ell} \frac{n_i}{1-p_{\mathrm{fail}}^{(i-1)}} \\
=& \prod_{i=1}^{\ell} \frac{1}{1-p_{\mathrm{fail}}^{(i)}}\frac{n_i}{n_i-2}\\
\le& \frac{1}{(1-1/2000)^4}\prod_{i=2}^{\infty} \left(1-\frac{1}{34} \left(544/2000\right)^{2^{i-1}}\right)^{-(2i)^2}\\
&\times\prod_{i=1}^{\infty}\frac{(2i)^2}{(2i)^2-2}\\
&\le 3
  \end{align*}
  where the second-to-last set of lines can be numerically evaluated to be approximately \(2.91\).
  Since \(f(x) = 1/x\) is convex on \((0,\infty)\) and \(N\) is in \((0,\infty)\) almost-surely, it follows from Jensen's inequality that
  \begin{align*}
      \expect{\frac{K}{N}} \ge \frac{K}{\expect{N}} \ge \frac{1}{3}
  \end{align*}

\end{proof}

\begin{proof}[Proof of property 3]
  Let \(K_{\ell} = \prod_{i=1}^{\ell} k_i\)  with \(K_0 = 1\), and we define \(M_{\ell}\) to be the memory consumed by \(\ell\)-concatenated QED distillation.
  To perform the \(\ell\)-concatenated QED distillation protocol, we must assemble a register of size \(n_{\ell} K_{\ell-1}\) on which to perform the outer distillation\footnote{Further memory optimizations should be possible, but we do not discuss them here.}.
  Furthermore, we must be capable of running the inner distillation protocol, which consumes \(M_{\ell-1}\) memory.
  Note that the inner distillation protocol produces the outputs in-place, so that we can re-use part of the outer distillation register to perform the inner distillation protocol.
 In particular, while the inner distillation protocol is being executed, at least \(K_{\ell-1}\) inputs to the outer distillation are missing, so
  \begin{align}
  \label{eq:memory}
    M_{\ell} &\le n_{\ell} K_{\ell-1} + M_{\ell-1} \\ \label{prop:concat-props:mem-upper1}
            &\le \sum_{i=1}^{\ell} n_i K_{i-1} \\
            &\le \sum_{i=1}^{\ell} 4^i \left(i!\right)^2\label{prop:concat-props:mem-upper2} \\
            &\le \ell 4^{\ell} \left(\ell! \right)^2 \le 4^{\ell} e^{2 \ell\log \ell} = e^{O(\ell \log \ell)}
  \end{align}
  For now, we take the loose upper bound \cref{prop:concat-props:mem-upper1}, which is equivalent to not optimizing the memory to be in-place.
  In \cref{prop:concat-props:mem-upper2}, we upper bound \(K_{\ell} \le \prod_{i=1}^{\ell} n_i\).

  As in property 1, to achieve a block error rate \(\epsilon\), we use the \(m\)-concatenated QED distillation with \(m = O(\log \log 1/\epsilon)\).
  We have that \(M_m = e^{O(\log \log 1/\epsilon\cdot \log \log \log 1/\epsilon)}\), so for some \(\alpha > 0\), we have the asymptotic upper bound
  \begin{align}
      M_m &\le \exp\left((\log \log \log 1/\epsilon)\cdot (\alpha \log \log 1/\epsilon)\right) \\
      &= O\left(\left(\log \log 1/\epsilon\right)^{\alpha \log \log 1/\epsilon}\right)
  \end{align}
\end{proof}

Note also that in practice, there may be further space-time trade-offs that modify the memory usage in a given setting.

We can combine the scheme of \cref{thm:constant-overhead} with a constant number of rounds of the BDSW-2EPP distillation sequence to achieve constant expected overhead using any input error rate less than \(1/2\):
\begin{corollary}[Main result]\label{coro:main-result}
  For physical Bell pairs subjected to a noise channel \(\mathcal{E}\) that is \(p\)-Pauli noise with \(p\in [0,1/2)\), there exists a family of concatenated error detecting distillation protocols parameterized by \(\ell\) satisfying the following properties
  \begin{enumerate}
  \item Conditioned on an output state being produced, the output contains no error with probability \(1-p_\ell\), where
    \begin{align}
      p_\ell \le \frac{1}{34} \left(\frac{544}{2000}\right)^{2^\ell}
    \end{align}
    In other words, to achieve an output block error probability \(\epsilon\), it suffices to pick \(\ell = O(\log \log 1/\epsilon)\).
  \item The protocol consumes \(N\) input Bell pairs to produce \(K\) output Bell pairs (repeat-until-success), where
  \begin{align}
    \expect{\frac{K}{N}} = \Omega(1)~.
  \end{align}
  \item For a desired block output error rate \(\epsilon\), the maximum memory required is \(O\left(\left(\log \log 1/\epsilon\right)^{\alpha \log \log 1/\epsilon}\right)\) for some \(\alpha >0\) which is quasi-polylogarithmic in \(1/\epsilon\).
  \end{enumerate}
\end{corollary}
\begin{proof}
    We use a constant number of rounds of the BDSW-2EPP distillation scheme (\cref{app:2EPP}) until the error rate falls below \(\frac{1}{2000}\), then use the scheme of \cref{thm:constant-overhead}.
    There are only a constant number of steps of the BDSW-2EPP distillation, so none of the asymptotics are changed.
\end{proof}

\begin{figure}[h!]
  \centering
  \includegraphics[width=\columnwidth]{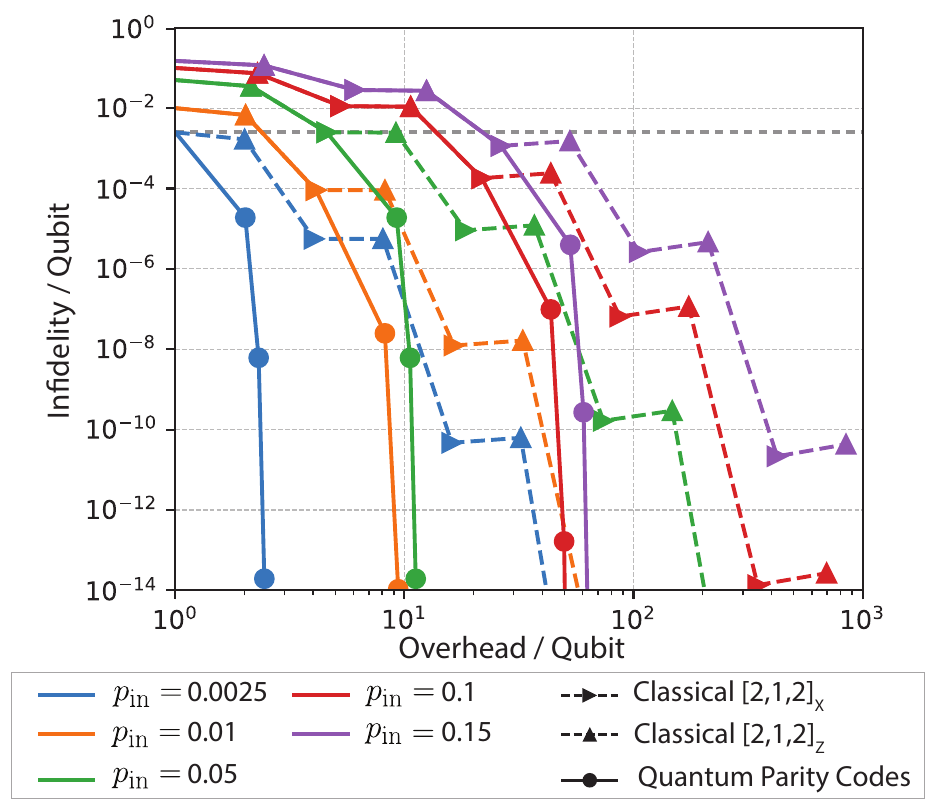}
  \caption{Concatenated QED distillation using the code sequence of \cref{thm:constant-overhead} along with initial rounds of a classical \([2,1,2]\) code. Each color corresponds to a different initial error rate: $0.0025$, $0.01$, $0.05$, $0.1$, and $0.15$. The dashed curves are the BDSW-2EPP distillation scheme, utilizing classical repetition codes $[2,1,2]$ alternately in the X and Z bases. The solid lines are the constant-overhead QED distillation scheme from \cref{coro:main-result}. The markers at each point indicate the basis used, as per the legend. Above the dashed horizontal line, we apply classical $[2,1,2]$ repetition codes (BDSW-2EPP scheme) to reduce the error before switching to quantum parity codes of growing size.
  }
  \label{fig:parity_2EPP}
\end{figure}

Figure \ref{fig:parity_2EPP} shows the infidelity per logical qubit as a function of physical Bell pair overhead per logical Bell pair for both methods: BDSW-2EPP (Definition \ref{defn:2epp-code-sequence}, dashed lines) and our quadratic parity code sequence (Definition \ref{defn:quadratic-code-sequence}, solid lines), combined with BDSW-2EPP to reach the initial threshold error rate. Different input physical error rates are shown in different colors. The asymptotic constant-rate behavior of our scheme is evident from the solid lines converging into a constant overhead (approaching a vertical line) for our scheme, while the conventional scheme has an overhead that continues to grow.

\section{Quantitative Analysis}
\label{sec:numerics}

In this section, we quantitatively analyze the performance of our asymptotically constant rate error-detection entanglement distillation scheme. We broaden our investigation beyond the codes that were used to prove constant rate, exploring more general code families. These families may not provide an asymptotically constant rate but offer favorable performance for finite sizes.

We assess the scheme's efficiency by examining distillation overhead and memory requirements, as well as the full operation time.
To do so, we first numerically identify the most efficient distillation sequences, optimizing over all possible code parameters for small classical codes and quantum codes based on Ref.~\cite{grassl2023codetables}.
We then compare the results against established methods such as BDSW-2EPP and lattice surgery, finding favorable performance.

Our techniques allow a continuous tradeoff between distillation overhead and buffer size---a term we define as the total logical qubit footprint used for networking on each node (see Figure \ref{fig:setup}). This tradeoff results in a significant reduction of almost an order of magnitude in both distillation overhead and logical gate operation time compared to other methods discussed here, such as BDSW-2EPP and lattice surgery, even with modest buffer sizes on the order of tens of logical qubits, thus demonstrating the practical applicability of our scheme.

\subsection{Distillation Optimization Methods}
\label{sec:optimization_methods}
\subsubsection{Model for Distillation Overhead}

We start by establishing the model used to optimize the distillation sequence. We implement a concatenated error-detection distillation scheme, as defined in Def.~\ref{defn:concat-ed-distillation}. This approach makes use of a series of QED codes \(\{\mathcal{C}_i\}_{i\in\mathbb{N}}\).
The initial level (0) directly produces noisy logical Bell pairs without distillation with a certain rejection probability $p_0$, typically following state injection of physical Bell pairs.
Subsequent levels ($i=1,..,\ell$) then employ QED codes $[[n_i, k_i, d_i]]$ and a repeat-until-success strategy, as described above.

Similar to Proposition \ref{prop:ed-distillation-error-rate}, we can upper bound the output error rate of a QED code $[[n_i,k_i,d_i]$ with an initial error rate $p_{i-1}$. Given a successful attempt, the output error probability $p_i$ is upper bounded by
\begin{align*}
    p_{i} &\le \Pr(\text{output error} \mid \lnot FAIL)\\
    &\le \frac{\Pr(|E| \ge d)}{\Pr(\lnot FAIL)} \\
    &\le \frac{\Pr(|E| \ge d)}{\Pr(|E| = 0)}  \\
    &\le \frac{1 - \sum_{j=0}^{d_i-1} \binom{n_i}{j} (p_{i-1})^j(1-p_{i-1})^{n_i-j}}{(1-p_{i-1})^{n_i}},
\end{align*}
where $|E|$ is the weight of the error, and we use Bayes' rule in step 1.

Similar to property 2 in Theorem~\ref{thm:constant-overhead}, the error detection overhead is estimated as follows. The QED code consumes $N_i$ input Bell pairs to produce $K_i$ output Bell pairs, with an effective overhead upper bounded by
\begin{align*}
\expect{\frac{N_i}{K_i}} \le \frac{1}{1 - p_{0}} \prod_{j=1}^{i} \frac{n_j}{k_j} \frac{1}{(1-p_{j-1})^{n_j}},
\end{align*}
where the first term accounts for the Bell state injection rejection rate $p_0$, and the last term accounts for the retry count of each distillation level.
At each level $i$, the output comprises $K_i = \prod_{j=1}^{i} k_j$ Bell pairs. Furthermore, according to Eq.~(\ref{eq:memory}), the cumulative memory usage up to and including step $i$ is:

\begin{align}
M_{i} &= \max\left(n_{i} K_{i-1}, (n_{i}-1) K_{i-1} + M_{i-1}\right).
\end{align}

Using this equation recursively, we find that
\begin{align}
(n_i-1)K_{i-1}+M_{i-1}
&\geq (n_i-1)K_{i-1} + n_{i-1}K_{i-2}\nonumber\\
&\geq (n_i-1)K_{i-1} + k_{i-1}K_{i-2}\nonumber\\
&= n_i K_{i-1}.
\end{align}
Therefore, the second term is always larger, and we can simplify the expression to
\begin{align}
M_i = (n_i-1)K_{i-1}+M_{i-1}
= \sum_{j=0}^{i-1}(n_{j+1}-1)K_j.
\end{align}

The final values for a full sequence $\{[[n,k,d]]\}_{i=1}^{\ell}$ of QED codes are then computed based on these parameters for $i=\ell$.
Note that this expression assumes that different stages of the distillation are performed sequentially, in order to minimize space usage. In practice, it may be desirable to pipeline the operations and perform some of them in parallel, as we discuss in Sec.~\ref{sec:pipeline_equations}.

\subsubsection{Distillation Sequence Optimization}
\label{sec:tree_search}
Using the models for memory usage and error rate from earlier subsections, we optimize the distillation overhead over distillation code sequences subject to the maximum buffer size ($M_{\text{max}}$), target output error rate ($p_{\text{target}} = 10^{-12}$), and the maximum number of distillation levels ($\ell_{\text{max}} = 7$).
We reduce the search space by performing a depth-first search and pruning branches when all distillation sequences in a branch are less optimal than the current most-optimal candidate or violate the buffer size constraint (see Alg.~\ref{alg:tree_search} in Appendix \ref{app:DFS_algorithm}).

We consider distillation sequences obtained from the following classical and quantum codes (\(\approx\) 500 codes)
\begin{itemize}
\item Quantum parity codes $[[n,n-2,2]]$, with \(n\le 40\).
\item Quantum Hamming codes $[[2^r,2^r-r-2,3]]$ with \( r\le 6\).
\item Best known QECCs from the code tables \cite{grassl2023codetables} as of January 2024 with \(n \le 30\).
\item Classical repetition codes $[n,1,n]$ for all \( n \le 12 \) in the \( X, Y\) and \( Z \) bases.
\end{itemize}

For classical codes, we use repetition codes $[n,1,n]$, so that we can directly calculate the rejection probability $p_{\text{fail}}$ and the error probability $p_{i}$ (See \cite{CodeLink}).
A classical code does not detect Pauli errors in one basis, so the probability of an error in each basis becomes unbalanced.
Thus, we include repetition codes in all three Pauli bases in the search.
Further optimization using unbalanced quantum codes is possible, but we exclude this optimization for simplicity.
Furthermore, we constrain our search such that classical codes may not be used after quantum codes in the distillation sequence.

\begin{figure*}
  \centering
  \includegraphics[width=1.9\columnwidth]{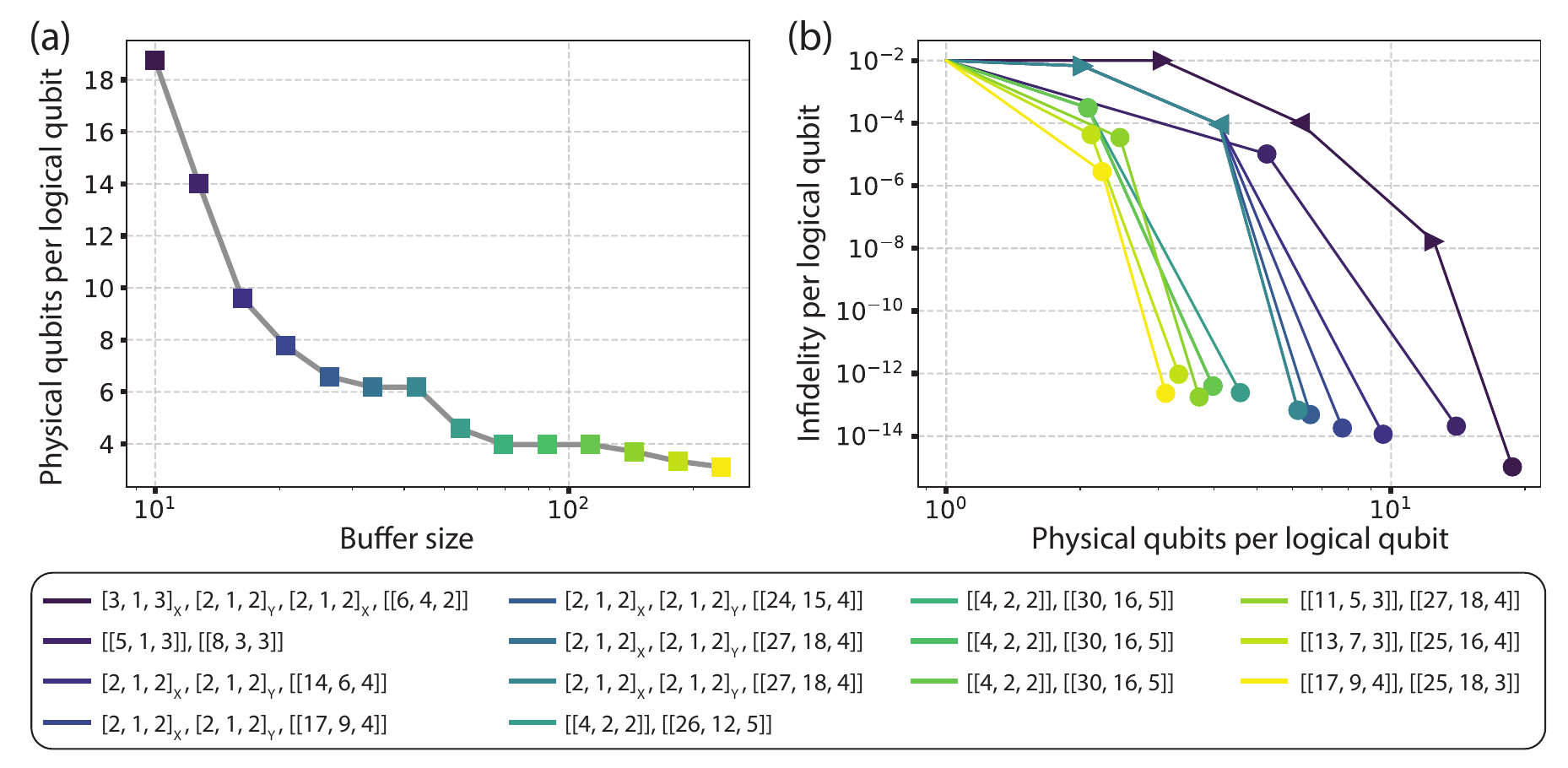}
  \caption{Optimized distillation sequences at a logical input Bell pair error rate of 1\%. (a) Distillation overhead per logical qubit as a function of the buffer size for a target logical Bell pair error rate of $10^{-12}$. Each point in (a) corresponds to a different optimized sequence presented in (b), color-coded for clarity. Quantum codes are represented by circles, and classical codes are represented by triangles, pointing upwards, to the right, or left, corresponding to the Z, X, and Y basis of the classical codes respectively. We also annotate the distillation sequence for each curve, with subscript denoting the basis for classical codes. Note that there is a permutation symmetry between X, Y, and Z for classical codes. Additionally, the state injection step is not considered here, as the input error rate represents the logical error rate after the injection step. 
  \label{sequence_p_in_1}
  }
\end{figure*}

\subsection{Distillation Optimization Results}
\label{sec:numerical_results}

We now show the results of our optimization, demonstrating low overheads of a few physical qubits per logical qubit, even with limited buffer sizes of tens of logical qubits.

In Fig.~\ref{sequence_p_in_1}, we present the optimized sequences designed for a Bell pair error rate of $1\%$. 
Fig.~\ref{sequence_p_in_1}(a) illustrates the distillation overhead for various buffer sizes, for a target logical Bell pair error rate of $10^{-12}$.
Corresponding to each point in this plot, the right plot details the sequence of codes used, with matching colors for clarity. Additional figures for various Bell pair error rates are provided in Appendix \ref{app:overheads_additional_Bell_errors}. 
As seen in Fig.~\ref{sequence_p_in_1}, we can achieve relatively low overheads, even with small buffer sizes. With a buffer size of around 50 logical qubits, our optimized distillation schemes can achieve as few as 4 physical Bell pairs per logical Bell pair, at $1\%$ input infidelity and output infidelity as low as $10^{-12}$. For slightly smaller buffers, e.g. 30 qubits, the overhead increases modestly to about 6.2 Bell pairs to maintain the same infidelity.
We observe that successful distillation sequences usually start with a small classical repetition code, before switching to a high-rate quantum code. For larger buffer sizes, starting with a high-rate quantum code is possible, and better overheads are achieved.

We next examine the trade-off between the Bell pair overhead and buffer size for a variety of Bell pair infidelities, as shown in Fig.~\ref{overheads_different_p_in}, with relevance to a wide variety of experimental settings with differing physical Bell pair error rates and buffer sizes.
Each data point in this figure results from the optimization process described in Sec.~\ref{sec:tree_search} and Fig.~\ref{sequence_p_in_1}. The optimized code sequence is shown in Fig. \ref{sequence_p_in_1} for Bell pair error of $1\%$, or Appendix \ref{app:overheads_additional_Bell_errors} for all other error rates.
As the input error rate increases, the distillation overhead increases as well, but not significantly. It is also interesting to note that for all error rates, the distillation overhead decays in roughly the same way as a function of buffer size.

\subsection{Estimating Practical Distillation Overhead and Comparison with Other Methods}

Thus far, we have focused on the setting of ideal local operations.
This can be approximately realized, up to the target error rate, by encoding the Bell pairs locally into error-correcting codes, thereby suppressing the local operation errors to a level that is negligible.
However, a careful comparison with other possible schemes---including those that directly operate an error-correcting code at the physical level---is necessary to understand the practical overhead implications.

In this section, we compare our distillation scheme with various other methods that have been proposed for quantum interconnects in large-scale distributed quantum computation, for realistic error models with both local gate errors and network noise.
This includes well-known distillation techniques like BDSW-2EPP~\cite{bennett1996purification}, as well as alternative approaches such as lattice surgery~\cite{fowler2010surface,ramette2023fault, sinclair2024fault}.

We examine the physical Bell pair overhead for a high fidelity distributed gate in each method. We conclude that our constant-rate distillation scheme provides a significant reduction in overhead compared to other methods with the same error model. We show results for a large parameters space - from small to large Bell pair error rates and memory sizes. For each parameter regime, we see a reduction in overhead, sometimes by orders of magnitudes. For example, for a Bell pair error rate of $1\%$, our overhead is $\sim 6$ Bell pairs, while the BDSW-2EPP scheme requires dozens and lattice surgery requires thousands.

\subsubsection{State Injection for Efficient Usage of Communication Interface}
\label{sec:injection}

\begin{figure}
  \centering
  \includegraphics[width=0.9\columnwidth]
  {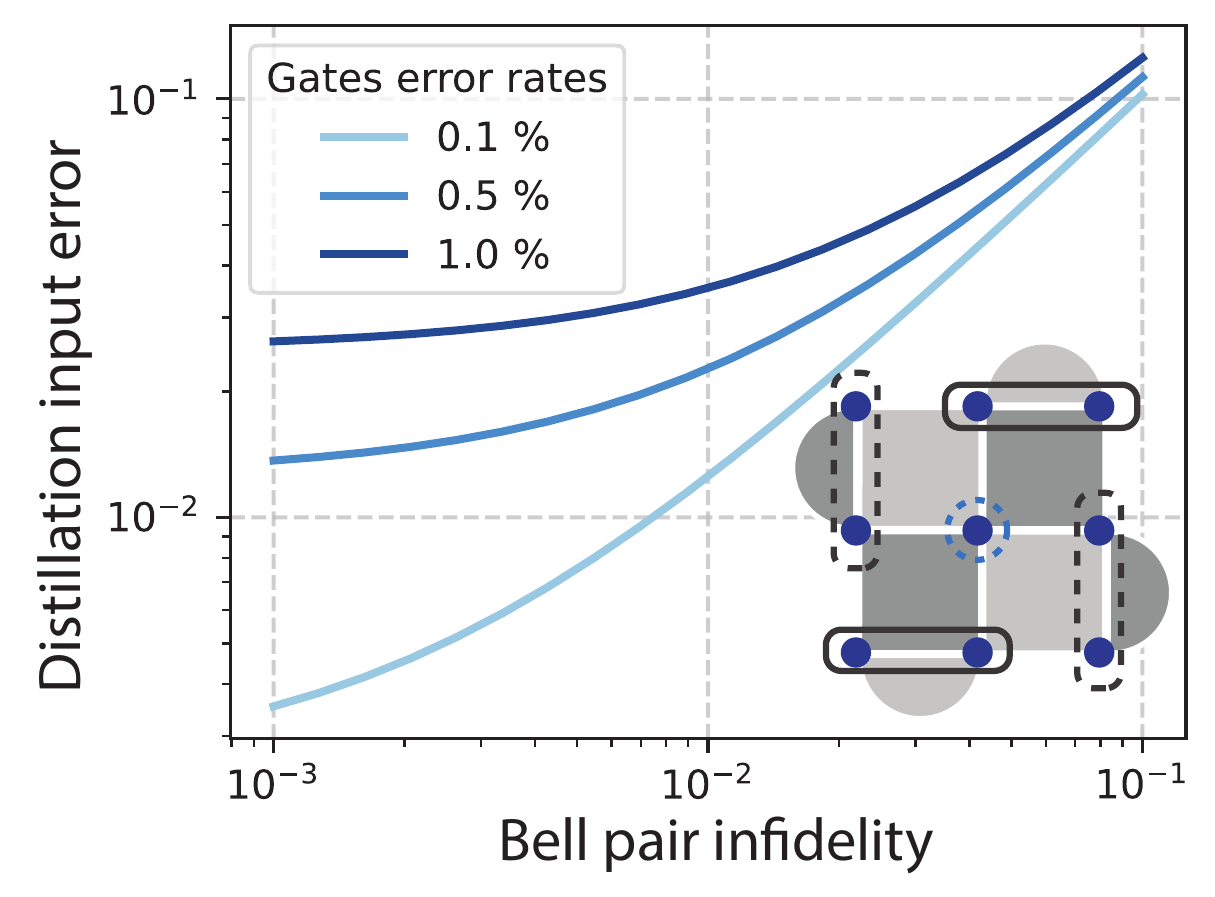}
  \caption{State injection with the Middle of the Rotated surface code approach (MR), as described in Ref.~\cite{lao2022magic}. We show the estimated distillation input error rate as a function of physical Bell pair error rate, for different local gate error rates, based on Eq.~(\ref{Eq:distill_injection_Bell}). Inset: a distance-3 logical qubit injection based on the rotated surface code. The blue-circled qubit is the magic state injected into the lattice. The other data qubits, circled with solid (dashed) black outlines, are initialized in $\ket{+}$ ($\ket{0}$) state. Ancilla qubits are not shown. Bright (dark) plaquettes represent X (Z) stabilizers for detecting Z (X) errors.}
  \label{fig:injection}
\end{figure}

In realistic experimental implementations, there will inevitably be both local gate errors and network errors.
Typically, many local gates can be executed in parallel, thereby delivering a much higher throughput.
It is thus commonly the case that the rate of gate execution over the network is much slower than local gate operations, and the former frequently becomes the bottleneck.
To make the best use of the rate-limiting quantum interconnect, we inject the distributed physical Bell pairs into rotated surface code logical qubits, thereby providing protection against local operation errors.
With modern state injection approaches~\cite{lao2022magic,li2015magic,gidney2023cleaner}, the injected logical error rate can be made comparable or even lower than the local physical gate error rate, thus not incurring much additional overhead.

We follow the approach of Ref.~\cite{lao2022magic}, where the state is injected to the qubit in the Middle of the Rotated surface code lattice (MR), while the other data qubits are initialized in either $\ket{+}$ or $\ket{0}$ states, as illustrated in Fig.~\ref{fig:injection}. The logical error rate of injecting a distance-\(d\) magic state using this approach is:
\begin{align}
\label{injection}
    p_{\text{injection}} = \frac{3}{5}p_2 + \frac{2}{3}p_1 + p_{\text{IN}} + O(p^2),
\end{align}
where $p_1$ and $p_2$ are the 1-qubit and 2-qubit gate error rates, respectively, and $p_{\text{IN}}$ is the initialization error, which is related to the physical Bell pair error rate here. The protocol is probabilistic, with a finite success rate which scales linearly with the error rate~\cite{lao2022magic}.

Upon successful injection of each qubit into a surface code on each node, we initiate the distillation process. The input error for distillation can thus be estimated as:
\begin{align}
\label{Eq:distill_injection_Bell}
p_{\text{distill,in}} = \frac{6}{5}p_2+\frac{4}{3}p_1+p_{\text{Bell}}.
\end{align}

A plot of $p_{\text{distill,in}}$ as a function of $p_{\text{Bell}}$ for different gate error rates $p_1=p_2$ is shown in Fig.~\ref{fig:injection}. One can see that when the network error is much larger than local gate errors, the state injection adds only a negligible amount of noise.

The injection process should be successful for both qubits in the Bell pair. Therefore, for a given injection rate $p_{\text{rejection rate}}$ which depends on the local gates errors, the Bell pair injection rate is:
\begin{align}
\label{Eq:injection_Bell_rejection}
p_{\text{Bell injection rejection rate}} = 1 - (1 - p_{\text{rejection rate}}) ^ 2.
\end{align}

We use the rejection rates calculated in~\cite{lao2022magic} for the numerical simulations presented in this paper. For a local gate error rate of $0.1 \%$ and $p_1 = p_2$ the injection rejection rate is $ \sim 8 \%$, setting the Bell injection rejection rate as $ \sim 15.36 \%$.

\subsubsection{Comparative Analysis of Network Overheads}

We now examine the physical Bell pair overhead necessary to distill a high quality logical Bell pair between two quantum nodes, using different quantum error correction schemes. We fix a local gate error rate of $0.1\%$ and evaluate how each scheme performs when we vary the network error rate. 

In our analysis, we consider the following schemes: \\
\textbf{Constant-rate QED distillation sequences:} These are the sequences we have developed in this work. \\
\textbf{BDSW-2EPP code sequence \cite{bennett1996mixed}:} Sequences of $[2,1,2]$ classical codes, as is typically used for entanglement distillation. \\
\textbf{Lattice surgery \cite{fowler2010surface, ramette2023fault}:} By distributing $O(d^2)$ physical Bell pairs along the edge connecting two surface code patches, one can perform a lattice surgery gate to distribute a logical Bell pair. \\

We note that there are other schemes that could be considered, e.g. using BDSW-1EPP with random quantum stabilizer codes~\cite{bennett1996mixed}.
However, to the best of our knowledge, this requires solving the potentially-challenging decoding problem for random quantum stabilizer codes, which is related to hardness assumptions for certain post-quantum cryptosystems~\cite{mceliece1978public,bernstein2017classic,melchor2018hamming,aragon2022bike}.
We therefore leave a detailed comparison such approaches to future work.

We use the results of Sec.~\ref{sec:numerical_results}, together with the Bell state injection errors in Sec.~\ref{sec:injection}, to estimate the quantitative overheads for the first two distillation methods.
The analysis of lattice surgery is adapted from Ref.~\cite{ramette2023fault, sinclair2024fault} and described in more detail in Appendix~\ref{app:lattice_surgery}.

We summarize our results in Table. ~\ref{table:comparison}.
Given Bell pair error rates ranging from $0.1\%$ to $15\%$, we tabulate the number of physical Bell pairs per logical Bell pair for each scheme.
Our constant-rate distillation scheme shows a remarkable reduction in resource overhead compared to other schemes.
Even with a buffer size of 30 logical qubits, we find almost an order of magnitude reduction in resource overhead cost compared to other competing schemes.
For example, given a Bell pair error rate of $1\%$, our scheme achieved overheads on the order of 6, while BDSW96 uses dozens of Bell pairs and lattice surgery employs thousands.

Additionally, our techniques allow a continuous trade-off between memory size and Bell pair rate, enabling tailored optimization for a given operational regime. Large buffer sizes allow for fewer than 10 physical Bell pairs per logical Bell pair, and smaller buffer sizes only lead to a modest increase of the network overhead. We show an example of buffer size of just 8 logical qubits in Appendix~\ref{app:small_buffer}. 

Furthermore, the network overheads presented here can be further optimized using known techniques, e.g. tailoring the distillation sequence to fit specific error models with bias~\cite{bonilla2021xzzx} or erasure errors~\cite{wu2022erasure,kubica2023erasure}.

We calculate a concrete upper bound for the physical memory overhead in Appendix~\ref{app:memory_overhead}, specializing to the case of injecting states into surface codes and operating with error rates around $10^{-12}$. To minimize memory overhead, various strategies can be employed in the future. For example, we can adapt code distances during the sequence, using smaller distances in initial levels where error rate demands are less strict. Another approach is to use high-rate quantum low-density-parity-check (QLDPC) codes for logical encoding rather than surface codes. Due to the concatenated nature of our scheme, where many operations occur in parallel at later stages, utilizing QLDPC codes can significantly enhance performance.

\begin{table*}[h]
    \centering
    \begin{tabular}{| c || c | c | c | c | c|}
    \hline
        \diagbox{Scheme}{Network error rate} & 0.1\% & 1\% & 5\% & 10\% & 15\% \\
         \hline\hline
         Distillation input error rate & 0.35\% & 1.25\% & 5.2\% & 10.2\% & 15.2\% \\\hline
         \hline
         BDSW-2EPP scheme - overhead & 75.82 & 78.89 & 348.75 & 411.58 & 2000.61 \\
        \hline
        BDSW-2EPP scheme with Y basis - overhead & 38.15 & 39.09 & 175.02 & 205.02 & 491.53 \\
        \hline
        Constant-overhead distillation (buffer = 10) - overhead & 8.98 & 22.44 & 55.51 & 102.75 & 188.96 \\
        \hline
        Constant-overhead distillation (buffer = 30) - overhead & 5.55 & 7.32 & 16.53 & 38.97 & 74.40 \\
        \hline
        Constant-overhead distillation (buffer = 50) - overhead & 3.90 & 7.32 & 16.53 & $\leq$ 30.71 & $\leq$ 67.32 \\
        \hline
        Constant-overhead distillation (buffer = 100) - overhead & 2.95 & 5.20 & $\leq 12.99$ & $\leq$ 27.16 & $\leq$ 67.32 \\
        \hline \hline
         Lattice surgery - overhead & 1,089 & 1,369 & 5,329 & 22,201 & 142,129 \\
         \hline\hline
    \end{tabular}
    \caption{Physical Bell pair overhead per logical gate required for different schemes for distributed quantum computation. The values in the table are calculated for gate error rate $p_{\text{gate}}=0.1\%$ and target output error rate $p_{\text{target}} = 10^{-12}$. With the exception of lattice surgery, we first use state injection with a rejection probability of $15.36 \%$, increasing the input Bell pair fidelity from the network error rate to the distillation input error rate. Values with an inequality represent an upper bound on overhead based on the existence of a specific optimized sequence, but a better sequence with lower overhead might exist.}
    \label{table:comparison}
\end{table*}

\subsubsection{Practical Considerations and Distillation Pipeline}
\label{sec:pipeline_equations}

In this section, we further elaborate on various trade-offs involved in implementing the distillation pipeline in practice.
In addition to the communication overhead and distillation throughput, we also discuss how the different stages of distillation can be constructed to make best use of the Bell pair and buffer memory resources.

To ensure that all distillation resources are put to effective use, it is desirable to balance the resource usage at different levels of distillation, such that no individual level bottlenecks the distillation process.
At the same time, in order to make full use of all incoming physical Bell pairs, the throughput of the distillation procedure should be at least as high as the rate of physical Bell pair generation.
Otherwise, some of the physical Bell pairs must be discarded, leading to a higher effective communication overhead.

We analyze these considerations in more detail in Appendix.~\ref{app:Distillation_pipeline_in_depth_overview}, and summarize the results here.
We build the process as a pipeline (Fig.~\ref{fig:example_pipeline_uniform}), where each level begins as soon as there are enough qubits (a logical block) output from the previous level.
For the state injection step, we require sufficient space to ensure that each newly arrive Bell pair has space for performing the injection, thereby requiring space usage
\begin{align}
B_0=\left\lceil \frac{T_{\text{inject}}}{T_{\text{Bell}}}\right\rceil,
\end{align}
in units of logical qubits.
For the $i$th distillation level, we again need to provide sufficient space to match the distillation throughput.
Considering the distillation time $T_{\text{distill},i}$ and the time it takes $T_{\text{input},i}$ to accumulate the input physical Bell pairs for one execution of the $i$th distillation level, we find that the space usage is
\begin{align}
B_i=\qty(\frac{T_{\text{distill},i}}{T_{\text{input},i}}+1)n_iK_{i-1},
\end{align}
where the latter factor $n_iK_{i-1}$ accounts for the space needed to store the outputs of previous levels of distillation, see Fig.~\ref{fig:concat-distillation-process}.
In Appendix.~\ref{app:Distillation_pipeline_in_depth_overview}, we further analyze $T_{\text{distill},i}$, finding that with an in-place parallel un-encoding circuit that minimizes space usage, we can upper bound $T_{\text{distill},i}\leq (3n_i-2-k_i)T_{\text{gate}}$, with $T_{\text{gate}}$ the typical logical entangling operation time.
Meanwhile, the average time required to accumulate the physical Bell pair inputs can be estimated as
\begin{align}
T_{\text{input},i}=n_i\prod_{j=0}^{i-1} \frac{n_j}{1-p_{\text{fail},j}},
\end{align}
with $n_0=1$.
The total space usage is then the sum of the space usage at all stages
\begin{align}
B_{\text{all}}=\sum_{i=0}^l B_i,
\end{align}
and we produce $\prod_{j=1}^l k_l$ Bell states every $T_{\text{Bell}}\prod_{j=0}^l\frac{n_i}{1-p_{\text{fail},i}}$.
Notice that the average time it takes to generate a high fidelity logical Bell pair (and therefore that of a distributed logical gate) is longer by the physical Bell pair distribution time by exactly a factor of the communication overhead.

We have focused on constructing the distillation pipeline in a way that minimizes additional space usage, since we expect local gate operations to be relatively fast compared to entanglement distribution.
For other operating regimes, many space-time tradeoffs are possible in order to obtain optimal usage of space and time at each level of the distillation.
For example, for a given distillation step, further optimizations can be made to reduce $T_i(n_i)$ from $O(n_i T_{\text{gate}})$ to $O(T_{\text{gate}})$ by employing ancilla qubits for circuit parallelization~\cite{fowler2012timeoptimal}.
It is also worth noting that the output of the distillation factory comes in batches: $K_l=\prod_{i=1}^l k_l$, so further storage may be required for consumption, with possible impacts on algorithmic compilation of distributed computation.
We leave the detailed optimization of these considerations for specific distributed quantum computing systems to future work.

\section{Application to Distributed Quantum Computation}
\label{sec:distributed_computation}
We now turn to a heuristic analysis of the implications of our results on distributed quantum computing.
While a complete analysis is beyond the scope of this manuscript, we nonetheless identify regimes of operation where our scheme can alleviate or remove a significant bottleneck in distributed quantum computing.

The distributed quantum computation can be influenced by a number of different parameters, including the size of each network module, the fraction of qubit resources available for use in communication, as well as remote entanglement fidelity and local logical operation fidelity.
For the purpose of our analysis, however, we only need to focus on the quantities derived from such an architecture.
For remote entanglement, this includes the time $t_e$ to generate a physical Bell pair between a pair of nodes and the number of physical Bell pairs $\alpha$ needed to generate a high quality logical Bell pair.
As noted above, the overhead factor $\alpha$ will depend on the available workspace used for entanglement distillation.
For local operations, we focus on the local operation time $t_{intra}$, as characterized by the typical timescale of a single logical operation implemented by e.g. lattice surgery, braiding, or transversal gates.

\begin{figure*}
\includegraphics[width=2\columnwidth]{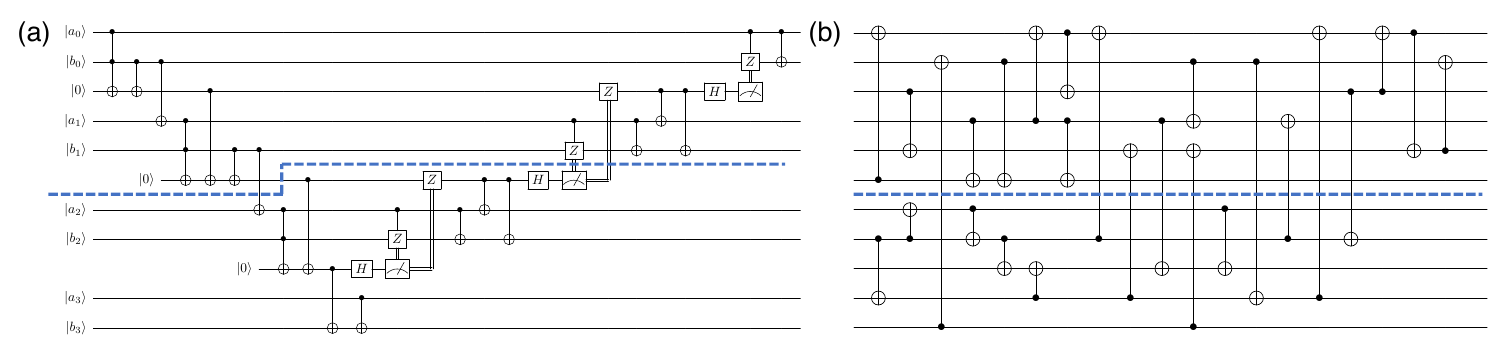}
\caption{Illustration of different quantum circuits and the communication cost of distributing them.
The number of CNOT or qubit lines the blue cut crosses is a good proxy for the number of Bell pairs that need to be distributed.
(a) Gidney adder circuit, where only two Bell pairs are required. Note that the uncomputation is done via measurement-based techniques, and therefore only requires classical communication.
(b) A random quantum circuit consisting of random CNOT pairs, which requires a larger number of distributed Bell pairs for a circuit with similar complexity, 9 for this particular cut.}
\label{fig:distributed_computation}
\end{figure*}

In addition to these system characteristics, another key determining factor of system resource tradeoffs is the algorithm itself.
Certain algorithms, such as the ripple carry adder, only require linear connectivity, thereby reducing the amount of intercore communication required.
Other quantum circuits, such as random quantum circuits, require much higher connectivity between qubits and therefore an extensive amount of intercore communication.
Since we are focusing on fault-tolerant quantum computers, the circuit compilation will also likely make heavy use of non-Clifford resource states such as T states and CCZ states, which require dedicated magic state factories on a surface-code-based architecture.
We characterize the algorithmic requirements of intercore communication with the parameter $\beta$, which we define to be the average number of intercore logical entanglement operations per core required for each intracore logical circuit layer.
Note that we choose to compare against circuit layers rather than circuit operations within each core, allowing for higher parallelism of operations within each core.

With this definition in mind, we can estimate $\beta$ for different circuits.
For the ripple carry adder executed on multiple cores, each with $s_c$ computational logical qubits, we may expect $\beta_{RCA}\approx \frac{1}{s_c}$, as illustrated in Fig.~\ref{fig:distributed_computation}, since only a single qubit and a single teleported gate need to be transferred from one register to the next for the full circuit of depth $2s_c$.
We note that this estimation may be modified if the circuit is executed in parallel using time-optimal quantum computation techniques~\cite{fowler2012timeoptimal}, or if oblivious carry runways are used to separate different sections of the adder~\cite{gidney2019approximate}.
When further accounting for the fact that non-Clifford gates are required and likely need to be prepared in dedicated modules with sizable spatio-temporal footprint (for example, the full factory in Ref.~\cite{gidney2019efficient} requires a spatial footprint of 12x6 logical qubits), we arrive at $\beta_{RCA}\approx 1$, since a single three-qubit CCZ resource state needs to be teleported every 3 circuit layers.
In contrast, for a random quantum circuit, we may expect a much higher $\beta_{RQC}= O(s_c)$, since each circuit layer involves gate operations on each logical qubit, and the pair of interacting logical qubits may be housed in different modules.

With these considerations, we can identify when the network interface becomes a bottleneck for the computation.
One layer of intramodule logical operation takes time $t_{intra}$, while an intermodule logical operation takes time $t_e \alpha$.
Thus, when
\begin{align}
\beta t_e\alpha \geq t_{intra},
\end{align}
the network communication becomes the limiting factor.
Depending on the entanglement distillation scheme employed, the distillation overhead factor $\alpha$ will be different.

For many current quantum architectures, $t_{intra}\ll\alpha t_e$, i.e. the logical clock speed within a node is much faster than the logical entanglement generation rate, given the parallel gate application within a module.
In the common case discussed above where non-Clifford gates are teleported in from dedicated factory modules, $\beta\gtrsim 1$, meaning that the communication interface will become a bottleneck.

For example, superconducting qubits are expected to achieve QEC cycle times of around 1 $\mu$s, translating to lattice surgery operation times of 10-30 $\mu$s depending on the code distance employed~\cite{acharya2022suppressing,krinner2022realizing}.
In contrast, the physical Bell pair generation rate via microwave-to-optical conversion with current estimates is expected to be around 1 MHz with a large, $15-20\%$ infidelity~\cite{ang2022architectures}.
Accounting for the distillation overheads in Tab.~\ref{table:comparison} for existing schemes shows that the logical entanglement overhead is around 500, resulting in one Bell pair every 500 $\mu$s, much slower than the intracore operations, posing a severe bottleneck.
The use of our constant-overhead distillation schemes, even with a buffer size as low as 30 logical qubits, is sufficient to lower this overhead by an order of magnitude, resulting in a 60 $\mu$s logical entanglement rate that is much closer to the local operation clock cycle.

As another example, neutral atom array systems have been able to achieve QEC cycle times of around 1 ms, translating to a similar timescale if transversal gates~\cite{shor1996fault,dennis2002topological} and algorithmic fault tolerance~\cite{cain2024correlated,zhou2024algorithmic} are employed.
The highest remote entanglement generation rate that has been shown in atomic systems, on the other hand, generates a single Bell pair every 5 ms with fidelity 94$\%$, such that the intranode \textit{logical} gate speed is faster than the internode \textit{physical} entanglement generation speed.
Our techniques bring down the distillation overhead from around 150 to 15, reducing the communication bottleneck depending on the algorithm of interest.

Therefore, in many cases of high relevance, our scheme can alleviate or remove the communication bottleneck in distributed quantum computing. These considerations also highlight an important tradeoff between different distillation schemes: when the communication rate is the main bottleneck, it is more important to use the high-rate communication schemes described here; when instead the operations within a given module are the bottleneck, it may instead be preferable to use communication schemes with lower memory overhead and employ more of the logical qubit resources towards active computation. These tradeoffs are also directly captured in our analysis by considering different memory buffer sizes.

\section{Conclusion}
\label{sec:conclusion}
In this paper, we proposed efficient entanglement distillation methods based on constant-rate distillation protocols and demonstrated their superior performance in regimes of interest, thereby addressing the communication bottleneck in fault-tolerant distributed quantum computing.
Using two-way communication and a sequence of quantum error-correcting codes with increasing rate, we achieve constant-rate entanglement distillation with much lower communication overheads than prior approaches, without needing to perform the hard decoding problem associated with conventional one-way QEC-based communication schemes.
The use of a first level of logical qubit encoding further ensures that we can ignore local operation errors, and make best use of remote Bell pairs that are typically generated at a lower rate and worse fidelity.
We find that our scheme can speed up the logical entanglement generation rate by an order of magnitude in relevant regimes, with a modest buffer memory size, and can be continuously adjusted depending on various space-time tradeoffs.
Our methods are also directly compatible with recent methods to speed up fault-tolerant computation on local nodes, based on transversal gates and algorithmic fault tolerance~\cite{cain2024correlated,zhou2024algorithmic}.
We thus believe that the techniques we introduced will be of key importance to future large scale quantum computers, where the computation may be distributed across multiple modules.

While we have proved that our scheme achieves constant rate and demonstrated its practical advantage in relevant parameter regimes, experimental implementations will likely benefit from a number of further optimizations.
For example, to reduce local resource consumption, lower code distances can be used at earlier stages of the distillation pipeline, since the error rate requirements are less stringent.
Alternatively, more efficient error-correcting codes such as high-rate quantum low-density-parity-check codes could be used.
This is particular useful here, considering that at later stages of the distillation pipeline, parallel operations are applied on entire blocks to prevent undesired error correlations.
The noise in such systems, whether across the network or within local nodes, may also have additional structure such as noise bias or erasures, which could be utilized to improve performance.
It will also be interesting to generalize these results to multipartite entanglement distribution.
Finally, we expect that the optimal choice of entanglement distillation scheme will also depend on the target quantum algorithm to be implemented and the way it is distributed across multiple nodes.
Although we initiated our study of this through simple estimates for a few representative algorithms, co-developing algorithms, compilation, and networking distribution schemes will likely be a fruitful area of future research.

\section{Acknowledgements}
We acknowledge helpful discussions with M.~Beverland, D.~Bluvstein, B.~Gu, L.~Jiang, B.~Li, J.~Ramette, J.~Sinclair, J.~Slote, S.~Wang, V.~Vuletic, Q.~Xu. 
We acknowledge financial support from IARPA and the Army Research Office, under the Entangled Logical Qubits program (Cooperative Agreement Number W911NF-23-2-0219), the DARPA IMPAQT program (grant number HR0011-23-3-0012), DARPA MeasQuIT program (grant number HR0011-24-9-0359), U.S. Department of Energy Office of Science (grant number DE-SC0020290), the National Science Foundation (grant number PHY-2012023), the QuSEC (grant number OMA-2326787), the CQN (grant number EEC-1941583), and the DOE/LBNL QSA (grant number DE-AC02-05CH11231).
G.B. acknowledges support from the MIT Patrons of Physics Fellows Society.
J.P.B.A. acknowledges support from the Generation Q G2 fellowship and the Ramsay Centre for Western Civilisation.
C.A.P. thanks QuEra Computing and Harvard University for hospitality, during which the majority of this work was completed.
The Institute for Quantum Information and Matter (IQIM) and the Center for Ultracold Atoms (CUA) are NSF Physics Frontiers Centers.

Note added: After the completion of this project, we became aware of related recent work~\cite{shi2024stabilizer}.
\printbibliography

\appendix
\onecolumn
\section{Statistics and Distributions}
\label{app:distributions}
Here, we provide some definitions and results about random variables that we utilize in the main text.
For a random variable (rv) \(X\), the corresponding probability density function (PDF) is denoted by \(f_X(x)\), and the corresponding cumulative distribution function (CDF) is denoted by \(F_X(x)\).
Define \(\overline{F}_X(x) = 1-F_X(x)\).

\subsection{Distributions}
We now fix the parameterization of some common distributions: For \(p \in (0,1)\), the geometric distribution, denoted \(\mathrm{Geom}(p)\), is the number of repeated independent attempts with success probability \(p\) before success is achieved.
For \(X \sim \mathrm{Geom}(p)\), the corresponding CDF supported on the positive real numbers \(\mathbb{R}_+\) is \(F_X(x) = 1-(1-p)^{\lfloor x \rfloor}\).
Likewise, for \(\lambda > 0\), we denote the exponential distribution as \(\mathrm{Exponential}(\lambda)\).
For \(X \sim \mathrm{Exponential}(\lambda)\) the corresponding CDF supported on \(\mathbb{R}_+\) is \(F_X(x) = 1-e^{-\lambda x}\).
For a random variable \(X \sim \mathrm{Gamma}(\alpha, \beta)\), its CDF is \(F_X(x)=\int_0^{\beta x} t^{\alpha-1} e^{-t} dt / \Gamma(\alpha)\).
Note that \(\mathrm{Gamma}(1,\lambda) = \mathrm{Exponential}(\lambda)\).
For \(p \in [0,1]\), a random variable \(X \sim \mathrm{Bernoulli}(p)\) is distributed according to \(f_X(1) = p\), \(f_X(0) = 1-p\).

We would like to model the distribution of the number of times our distillation scheme must consume a physical level Bell pair.
The number of times the distillation process must run in order to produce some number of outputs is modeled by a negative binomial distribution.

\begin{definition}[Negative binomial distribution]
  Fix two parameters \(r \in \mathbb{N}\) and \(p \in (0,1)\), and consider a bitstring with each bit drawn from \(\mathrm{Bernoulli}(p)\) i.i.d. The negative binomial distribution, denoted \(\mathrm{NB}(r, p)\), counts the number of \(0\) entries from the start of the bitstring before \(r\) \(1\) entries are seen.
  For an integer \(n\ge 0\), the distribution of \(X \sim \mathrm{NB}(r,p)\) has the probability mass function
  \begin{align}
    \Pr(X=n) = \binom{n+r-1}{n} p^n(1-p)^{r},
  \end{align}
  which is the same as the probability that exactly \(r\) successes are seen in \(n+r\) i.i.d Bernoulli trials such that the last trial is a success.
\end{definition}

In the concatenation procedure, each attempt at distillation in level-\(\ell\) requires that we invoke the distillation protocol at level-\((\ell-1)\).
Thus, we will need to consider a sum where the number of terms is a random variable and the terms are i.i.d. random variables.
Such distributions are known as compound distributions.
In particular, we will use the compound negative binomial distribution:
\begin{definition}[Compound negative binomial (CNB) distribution]
  Let \(N\sim \mathrm{NB}(r, p)\) be a negative binomially distributed random variable, and let \(\{Y_{i}\}_{i=1}^{\infty}\) be an infinite set of random variables distributed i.i.d. according to some distribution \(\mathcal{Y}\).
  Then, the random variable
  \begin{align}
    X = \sum_{i=1}^N Y_i
  \end{align}
  is said to be distributed according to a compound negative binomial (CNB) distribution.
  We will refer to the distribution of \(X\) as a \((r, p, \mathcal{Y})\)-CNB distribution.
\end{definition}
Compound distributions are used in insurance risk modeling where there are a random number of claims each with a random payout.
We are interested in a tail bound, i.e. upper bounding the probability that some extreme value is exceeded.
Our main tool will be that of \emph{stochastic dominance} which allows us to replace a random variable with one that always takes on a more extreme result.

\subsection{Stochastic dominance}
In this subsection, we include some tools that we expect to be useful to prove a parallel-repetition type result for constant-rate entanglement distillation schemes with post-selection, e.g. by outputting a maximally mixed state when a large number of distillation failures occur. However, we leave the detailed analysis to further work: The contents of this subsection are unnecessary to prove constant expected rate and may be skipped.
\begin{definition}[Stochastic dominance]
  For two random variables \(X\), \(Y\) taking values on \(\mathbb{R}\), we can define a partial ordering known as stochastic dominance.
  We say that \(X \succeq Y\) or \(X\) stochastically dominates \(Y\), if for all \(t \in \mathbb{R}\), we have that
  \begin{align}
    \Pr(X \ge t) \ge \Pr(Y \ge t)
  \end{align}
  In other words, the probability that \(Y\) is more extreme than \(t\) is no more than the probability that \(X\) is more extreme than \(t\).
\end{definition}
In calculations, it is helpful to remember that \(Y \succeq X\) iff \(\forall t.~ F_{Y}(t) \le F_{X}(t)\).
The notion of stochastic dominance allows us to relate tail bounds for simple distributions to tail bounds of more complicated distributions.
The problem of upper bounding the tail of compound negative binomial random variables has been studied for the purposes of insurance pricing.
The objective is to employ the result of \cite{willmot1997upper} by alternately employing tail bounds and replacing random variables by new random variables that stochastically dominate them.
Intuitively, the sum of a fixed number of Gamma distributed (exponentially distributed being a special case) random variables is itself Gamma distributed, and the additional complication of a random number of terms of the sum does not change the answer much.
\begin{fact}[\cite{willmot1997upper} Corollary 3]\label{fact:willmot-lin}
  Fix a distribution \(\mathcal{Y}\), \(q \in (0,1)\), \(r \in \mathbb{N}\), and let \(X\) be distributed according to a \((r,q,\mathcal{Y})\)-CNB distribution.
  Let \(Y \sim \mathcal{Y}\) and define \(\kappa\) to be the solution to
  \begin{align*}
      \int^\infty_0 e^{\kappa t} f_Y(t) dt = \frac{1}{q}
  \end{align*}

  Then, for \(Z \sim \mathrm{Gamma}(r, \kappa)\), \(Z \succeq X\).
\end{fact}

We begin by proving that stochastic dominance is preserved by addition.
This will allow us to prove that two CNB distributed random variables will stochastically dominate when the random variables in the sum stochastically dominate.
\begin{lemma}\label{lemma:stoch-dom-sum}
    For random variables \(X_1, X_2, Y_1, Y_2\) such that \(Y_1 \succeq X_1\) and \(Y_2 \succeq X_2\), their sums also stochastically dominate:
    \begin{align*}
        Y_1 + Y_2 \succeq X_1 + X_1
    \end{align*}
\end{lemma}
\begin{proof}
    The distribution of the sum of two random variables is given by convolution:
    \begin{align*}
        F_{X_1+X_2}(a) &= \int^a_{\infty} \int^{\infty}_{-\infty} f_{X_1}(t-u) f_{X_2}(u) \mathrm{d}t \mathrm{d}u\\
        &= \int_{-\infty}^\infty F_{X_1}(a-u) f_{X_2}(u) \mathrm{d}u\\
        &\ge \int_{-\infty}^\infty F_{Y_1}(a-u) f_{X_2}(u) \mathrm{d}u\\
        &= \int_{-\infty}^\infty f_{Y_1}(u) F_{X_2}(a-u) \mathrm{d}u \\
        &\ge \int_{-\infty}^\infty f_{Y_1}(u) F_{Y_2}(a-u) \mathrm{d}u = F_{Y_1+Y_2}(a)
    \end{align*}
\end{proof}

\begin{lemma}
    Fix \(r \in \mathbb{N}\), \(p \in (0,1)\), and two distributions \(\mathcal{X}\), \(\mathcal{X}'\) such that for independent random variables \(X \sim \mathcal{X}\), \(X' \sim \mathcal{X}'\), \(X \preceq X'\).
    Let \(Z\) (\(Z'\)) be an independent \((r,p,\mathcal{X})\)-CNB distributed ( \((r,p,\mathcal{X}')\)-CNB distributed) random variable.
    Then \(Z \preceq Z'\).
\end{lemma}
\begin{proof}
    Let \(N \sim NB(r,p)\) (\(N'\)) be the negative binomial random variable serving as the upper bound of the sum in the definition of \(Z\) (\(Z'\)).
    The conditional CDF \(F_{Z' | N}\) is a fixed sum of random variables, so we can use \cref{lemma:stoch-dom-sum} to lower bound the CDF for each value of the summation variable.
    \begin{align*}
        F_{Z'}(x) = \sum_{n = 0}^{\infty} f_N(n) F_{Z' | n}(x) \le \sum_{n = 0}^{\infty} f_N(n) F_{Z | N}(x) = F_Z(x)
    \end{align*}
\end{proof}

\section{Proofs omitted from \cref{sec:distillation_setup}}\label{app:distillation-proofs}

Recall that \(S_{\mathrm{ED}} = \{s_i\}_{i\in [r]}\) is a set of \(r\) independent stabilizer generators for an \(\dsl n, k, d \dsr\) stabilizer code with encoding map \(\mathcal{D}^{-1}\).
\propDistillCorrectness*

\begin{proof}
  Let \(\{L_{i,X}, L_{i,Z}\}_{i\in [k]}\) denote a basis of independent logical operators such that under the unencoding circuit \(\mathcal{D}\), \(i \in [k]\), \(\mathcal{D}^{-1} L_{i,X}\mathcal{D} = X_i\), \(\mathcal{D}^{-1} L_{i,Z}\mathcal{D} = Z_i\), and \(j\in [r]\), \(\mathcal{D}^{-1} s_{j}\mathcal{D} = Z_{k+j}\).
  Denote the \(r\) generators of \(S_\mathrm{ED}\) on the \(A\) and \(B\) subsystems as \(S_{\mathrm{ED},A} = \{s_i \otimes I\}_{i \in [r]}\) and \(S_{\mathrm{ED},B} = \{I \otimes s_i\}_{i \in [r]}\), respectively.

  The initial joint state is a stabilizer state with stabilizer generators \(S_{\mathrm{initial}}= \{X_i\otimes X_i, Z_i\otimes Z_i\}_{i \in [n]}\).
  The products of logical operators \(\mathcal{L}_{\mathrm{Bell}} = \{L_{i,X}\otimes L_{i,X}, L_{i,Z}\otimes L_{i,Z}\}_{i \in [k]}\)  and the operators \(S_{\mathrm{ED},\mathrm{Bell}}=\{s \otimes s\}_{s \in S_{\mathrm{ED}}}\) are therefore also stabilizers of the initial joint state with eigenvalues +1.

  \(\mathcal{L}_{\mathrm{Bell}}\) commutes with \(S_{\mathrm{ED},A}\) and \(S_{\mathrm{ED},B}\), so the post-measurement state after \cref{alg:distill:measurements} has the stabilizer generators \(S_{\mathrm{meas.}} = \mathcal{L}_{\mathrm{Bell}} \sqcup (-1)^{a} S_{\mathrm{ED},A} \sqcup (-1)^{b}S_{\mathrm{ED},B}\).
  \(S_{\mathrm{ED},\mathrm{Bell}}\) also commutes with \(S_{\mathrm{ED},A}\) and \(S_{\mathrm{ED},B}\), so the set \(\mathcal{L}_{\mathrm{Bell}} \sqcup (-1)^{a} S_{\mathrm{ED},A} \sqcup S_{\mathrm{ED},\mathrm{Bell}}\) is also a valid set of stabilizer generators for the state after \cref{alg:distill:measurements}, i.e. they both generate \(\langle S_{\mathrm{meas.}} \rangle\).
  We conclude that \(a+b=0\) i.e. \(a = b\).

  The recovery operator \(R_{AB}=R\otimes R\) sends
  \begin{align}
R_{AB}(-1)^{a}S_{\mathrm{ED},A}R_{AB} &= S_{\mathrm{ED},A} \\ R_{AB}(-1)^{b}S_{\mathrm{ED},B}R_{AB} &= (-1)^{a+b}S_{\mathrm{ED},B} \\ R_{AB}\mathcal{L}_{\mathrm{Bell}}R_{AB} &= \mathcal{L}_{\mathrm{Bell}}
  \end{align}
  so the post-recovery state after \cref{alg:distill:recovery} has stabilizer generators \(\mathcal{L}_{\mathrm{Bell}} \sqcup S_{\mathrm{ED},A} \sqcup S_{\mathrm{ED},B}\).
  After \cref{alg:distill:unencode}, the final unencoded state has stabilizer generators \(\{Z_i^{(A)}\otimes Z_i^{(B)}, X_i^{(A)}\otimes X_i^{(B)}\}_{i \in [k]}\sqcup \{Z_{k+j}^{(A)}\}_{j \in [r]} \sqcup \{Z_{k+j}^{(B)}\}_{j \in [r]}\), so the first \(k\) qubits are the desired Bell pairs.
\end{proof}

\propDistillErrorCorr*
\begin{proof}
  We proceed by a similar calculation to \cref{prop:correctness} with all stabilizer generators conjugated by \(E^{(B)}\).
  Let \(\sigma'\in \F^r\) be the syndrome of \(E^{(B)}\) i.e. \((-1)^{\sigma'}S_{\mathrm{ED},B}= E^{(B)} S_{\mathrm{ED},B} E^{(B)}\).
  The initial state has stabilizer generators \(E^{(B)}S_{\mathrm{initial}} E^{(B)}\), which include \(E^{(B)} \mathcal{L}_{\mathrm{Bell}}E^{(B)}\) and \((-1)^{\sigma'} S_{\mathrm{ED},\mathrm{Bell}} \).
  The operators \( S_{\mathrm{ED},B}\) are measured, but the stabilizer generators are \((-1)^{\sigma'} S_{\mathrm{ED},\mathrm{Bell}} \), so the measurement result is \(b=a+\sigma'\) i.e. \(\sigma' = \sigma\).
  Furthermore, the post recovery state after \cref{alg:distill:recovery} has stabilizer generators \(E^{(B)}\mathcal{L}_{\mathrm{Bell}}E^{(B)} \sqcup S_{\mathrm{ED},A} \sqcup (-1)^{\sigma}S_{\mathrm{ED},B}\).
\end{proof}

\propOutputErrorRate*
\begin{proof}
  Let \(E\) be the input error.
  Since the code is distance-2, for there to be an output and for the output to contain an error, we must have \(|E| \ge 2\).
  We can use this fact and Bayes' rule to compute bounds on the appropriate conditional probability
  \begin{align}
    p_{out} & \le \Pr(|E| \ge 2 \mid \lnot FAIL) \\
            & \le \frac{\Pr(|E| \ge 2)}{\Pr(\lnot FAIL)}\\
            & \le \frac{\Pr(|E| \ge 2)}{\Pr(|E| = 0)}\\ 
            &= (1-p)^{-n}\left(1-(1-p)^{n}-n p (1-p)^{n-1}\right) \\ 
            &= \left((1-p)^{-n}-1-n \frac{p}{1-p}\right) \label{prop:ed-distillation-error-rate:eq2}\\
            &\le \left(\exp\left[n \frac{p}{1-p}\right]-1-n \frac{p}{1-p}\right) \\
            &\le \left( n \frac{p}{1-p} \right)^2,
  \end{align}
  where we have used \cref{lemma:1mx-bound} to bound \cref{prop:ed-distillation-error-rate:eq2} and \(p < 1/(2n)\), \(n\geq 4\) \(\implies n \frac{p}{1-p} \le 1\) to apply \cref{lemma:expm1mx-quadratic} in the final line.
\end{proof}

\propErrorDetectionOverhead*
\begin{proof}
  The probability that the probabilistic protocol fails satisfies \(1-p_{\mathrm{fail}} \ge (1-p)^n\), so, by \cref{prop:retry-count}, the average attempt count \(T\) is \(\expect{T} \le (1-p)^{-n}\).
  The output is \(k=n-2\) Bell pairs and each attempt consumes \(n\) Bell pairs.
  \begin{align}
    \expect{\frac{N}{k}} &= \expect{T}\frac{n}{n-2}\\
    &\le (1-p)^{-n} \frac{n}{n-2}\\
    &\le (3n p + 1) \frac{n}{n-2},
  \end{align}
   where we have applied \cref{lemma:1mx-linear-bound} in the last inequality to simplify the bound to be linear in \(p\).
 \end{proof}

\section{BDSW-2EPP}
\label{app:2EPP}
Here, we provide a proof that the BDSW-2EPP scheme of \cite{bennett1996mixed} reduces the error rate for \(p\)-depolarizing noise with \(p \in (0, 1/2)\).
We omit an explicit proof that repeated application of this procedure results in an error rate converging to zero; however, the bound in \cref{thm:2-epp:error-reduction} is sufficient to imply this.

\begin{proposition}[BDSW-2EPP \cite{bennett1996mixed}]\label{prop:2-epp}
  Consider an error operator \(E\in \mathcal{P}^2\) with the support on each qubit drawn from \((p_X,p_Y,p_Z)\)-Pauli noise i.i.d.
  If \(p_Z \le p_X,p_Y\) and \(p_X+p_Y+p_Z < 1/2\), then BDSW-2EPP applied to \(E_B\ket{\Phi}_{AB}^{\otimes 2}\) with the check \(ZZ\) results in an output \(E' \ket{\Phi}_{AB}\) where \(E'\) is distributed according to \((p'_X,p'_Y,p'_Z)\)-Pauli channel where \(p'_X+p'_Y+p'_Z < p_X+p_Y+p_Z\).
\end{proposition}
\begin{proof}
  Define \(p_I= 1-p_X-p_Y-p_Z\).
  After checking and postselecting for bitflips with the operator \(ZZ\), the output is in the codespace, but potentially differs from the desired state by an operator \(L \in \mathcal{P}\).
  The proof proceeds by enumeration of the 16 events in the table.
  \begin{table}[h]
    \centering
    \begin{tabular}{|c|c|c|c|}\hline
      \(E\)  & \(\Pr(E)\) & FAIL  & \(L\)\\\hline
      \(II\) & \(p_I^2 \) & N & \(I\)\\\hline
      \(IX\) & \(p_Ip_X\) & Y & \\\hline
      \(IY\) & \(p_Ip_Y\) & Y & \\\hline
      \(IZ\) & \(p_Ip_Z\) & N & \(Z\)\\\hline
      \(XI\) & \(p_Ip_X\) & Y & \\\hline
      \(XX\) & \(p_X^2\)  & N & \(X\)\\\hline
      \(XY\) & \(p_Xp_Y\) & N & \(Y\)\\\hline
      \(XZ\) & \(p_Xp_Z\) & Y & \\\hline
      \(YI\) & \(p_Yp_I\) & Y & \\\hline
      \(YX\) & \(p_Xp_Y\) & N & \(Y\)\\\hline
      \(YY\) & \(p_Y^2\)  & N & \(X\)\\\hline
      \(YZ\) & \(p_Yp_Z\) & Y & \\\hline
      \(ZI\) & \(p_Zp_I\) & N & \(Z\)\\\hline
      \(ZX\) & \(p_Xp_Z\) & Y & \\\hline
      \(ZY\) & \(p_Yp_Z\) & Y & \\\hline
      \(ZZ\) & \(p_Z^2\)  & N & \(I\)\\\hline
    \end{tabular}
    \label{tab:2epp-events}
  \end{table}
  Summing rows, in which FAIL was not output, we define
  \begin{align}
    A &= \Pr(\lnot FAIL)\\
      &= p_I^2+p_X^2+p_Y^2+p_Z^2+2p_Ip_Z+2p_Xp_Y\\
      &= (p_I+p_Z)^2+(p_I+p_Z-1)^2
  \end{align}
  For \(\mathcal{O} \in \mathcal{P}\), let \(p'_{\mathcal{O}}=\Pr(L=\mathcal{O}|\lnot FAIL)\).
  The output has errors according to a \((p'_X,p'_Y,p'_Z)\)-Pauli noise channel with
  \begin{align*}
    \Pr(\lnot FAIL) &= A \\
    p'_I &= A^{-1}(p_I^2 + p_Z^2) \\
    p'_X &= A^{-1}(p_X^2+p_Y^2) \\
    p'_Y &= A^{-1}2p_Xp_Y \\
    p'_Z &= A^{-1}2p_Ip_Z \\
  \end{align*}

  We define \(q = p_I + p_Z\) and \(\bar{q}=1-q=p_X+p_Y\).
  If \(p_X + p_Y + p_Z < 1/2\) and \(p_Z \le p_X, p_Y\), then there exists \(\epsilon \in (0,1/2)\) such that \(q = \frac{1}{2} + \epsilon\) and \(\bar{q} = \frac{1}{2} - \epsilon\).
  The formulas then simplify to the following:
  \begin{align*}
    A &= q^2+\bar{q}^2 \\
    q'\equiv p_I' + p'_Z = \frac{q^2}{q^2+\bar{q}^2} &= \frac{1}{2} + \frac{2 \epsilon}{1+4\epsilon^2}\\
    \bar{q}'\equiv p_X' + p'_Y = \frac{\bar{q}^2}{q^2+\bar{q}^2} &= \frac{1}{2} - \frac{2 \epsilon}{1+4\epsilon^2}\\
    p'_Z        = \frac{2(q-p_Z)p_Z}{q^2+\bar{q}^2} &= 2 p_Z \frac{1- 2p_Z +2\epsilon}{1+4\epsilon^2}
  \end{align*}

  We proceed with the following lower bound using \(p_Z \le p_X, p_Y \implies p_Z \le 1/4 - \epsilon/2\).
  \begin{align}
      p_I'-p_I &= (q' - p_Z') - (q - p_Z) \\
      &= \left(\frac{2 \epsilon}{1+4\epsilon^2} - \left(2 p_Z \frac{1- 2p_Z +2\epsilon}{1+4\epsilon^2}\right)\right) - (\epsilon - p_Z) \\
      &= \frac{(p_Z-\epsilon)(4\epsilon^2 + 4p_Z -1)}{1+4 \epsilon^2} \\
      &\ge -\epsilon \frac{(4\epsilon^2 + 4p_Z -1)}{1+4 \epsilon^2} \\
      &\ge -\epsilon \frac{(4\epsilon^2 + (2\epsilon-1) -1)}{1+4 \epsilon^2} \\
      &= \frac{2\epsilon - 2 \epsilon^2 - 4 \epsilon^3}{1+4\epsilon^2} \label{thm:2-epp:error-reduction}\\
      &> 0
  \end{align}
\end{proof}

\section{Miscellaneous Bounds}
In this section, we include miscellaneous bounds that are needed in some of the proofs of the main results.
\begin{lemma}\label{lemma:1mx-bound}
  For \(n \ge 1\), \(x \in (0,1)\),
  \begin{align}
    (1-x)^{-n} \le \exp \left(n \frac{x}{1-x}\right)
  \end{align}
\end{lemma}
\begin{proof}
  From \(-\log(1-x) \le \frac{x}{1-x}\).
\end{proof}

\begin{lemma}\label{lemma:1mx-linear-bound}
  For \(x\in (0,1/n)\), and \(n\ge 4\)
  \begin{align}
    (1-x)^{-n} \le 3 n x + 1
  \end{align}
\end{lemma}
\begin{proof}
  The function \((1-x)^{-n}\) is convex on the interval \((-\infty,1)\), so we can use Jensen's inequality for the interval \((0,1/n)\): \(f(x) \le f(0) + x \frac{f(1/n) - f(0)}{1/n}\).
  \begin{align}
    (1-x)^{-n}  &\le 1 + x \frac{(1-1/n)^{-n}-1}{1/n}\\
               &\le 1 + \frac{175}{81} n x \\
               &\le 1 + 3 n x
  \end{align}
  where we have used that \((1-1/n)^{-n} \le \frac{256}{81}\) for \(n \ge 4\).
\end{proof}

\begin{lemma}\label{lemma:expm1mx-quadratic}
  For \(x \in [0,1]\),
  \begin{align}
    e^x-1-x \le (e-2)x^2
  \end{align}
\end{lemma}
\begin{proof}
  \begin{align}
    (e^x-1-x)-(e-2)x^2 &=\sum^{\infty}_{k=3}\frac{x^k}{k!} - \left(e-\frac{5}{2}\right)x^{2} \\
                       &\le\sum^{\infty}_{k=3}\frac{1}{k!} - \left(e-\frac{5}{2}\right)x^{2} \\
                       &= \left( e-\frac{5}{2}\right) (1-x^2) \\
                       &\le 0
  \end{align}
\end{proof}

\begin{lemma}\label{lemma:xlikelihood-sqr-upper}
  For \(\epsilon \in (0,1)\), \(x \in (0,\epsilon)\),
  \begin{align}
    \frac{x}{1-x} \le \frac{x}{1-\epsilon}
  \end{align}
\end{lemma}

\begin{lemma}\label{lemma:error-recurrence}
  For \(c\in \mathbb{R}\), \(m \in \mathbb{N}\), the recurrence relation \(a_{k}=c k^{2^m} a_{k-1}^2\) has the closed form solution
  \begin{align}
    a_k = \left(\prod_{i=1}^k i^{2^{k+m-i}}\right) c^{2^k-1}a_0^{2^k}
  \end{align}
\end{lemma}
\begin{proof}
  We proceed inductively.
  \begin{align}
    a_1 = \left(1^{2^{1+m-1}}\right) c^{2-1}a_0^{2} = c a_0^2
  \end{align}
  and
  \begin{align}
    a_k &= c k^{2^m} a_{k-1}^2 \\
        &= c k^{2^m} \left(\left(\prod_{i=1}^{k-1} i^{2^{k-1-i+m}}\right) c^{2^{k-1}-1}a_0^{2^{k-1}}\right)^2\\
        &= c k^{2^m} \left(\prod_{i=1}^{k-1} i^{2^{k-i+m}}\right) c^{2^k-2}a_0^{2^k}\\
        &= \left(\prod_{i=1}^k i^{2^{k-i+m}}\right) c^{2^k-1}a_0^{2^k}\\
  \end{align}
\end{proof}

\begin{lemma}\label{lemma:error-recurrence-product}
  For \(c \in \mathbb{R}_+\)
  \begin{align}
    \prod_{i=1}^k i^{2^{c-i}} \le 2^{2^c}
  \end{align}
\end{lemma}
\begin{proof}
  \begin{align}
    \log_2\left(\prod_{i=1}^k i^{2^{c-i}}\right) &= \sum_{i=1}^k2^{c-i}\log_2 i \\
                                                 &\le 2^c\sum_{i=1}^\infty 2^{-i}\log_2 i \\
                                                 &\le 2^c
  \end{align}
\end{proof}

\section{Modified Depth-First Search Algorithm For Finding Best Overhead Sequence}
\label{app:DFS_algorithm}

\begin{algorithm}
\caption{Find Best Overhead Sequence (Depth-First Search)}
\label{alg:tree_search}
\begin{algorithmic}[1]
\State \textbf{Require:} $\text{codes}, \text{max\_levels}, p_{\text{target}}, M_{\text{max}}, p_0$
\State \textbf{Initialize:} $\text{Viable\_Sequences} \gets []$
\State \textbf{Call:}  $\text{DFS\_search}(1, \text{current\_sequence}=[])$
\State \textbf{Return:} Sequence with best overhead in $\text{Viable\_Sequences}$
\\
\State \textbf{Procedure} $\text{DFS\_search}(i, \text{current\_sequence})$
\If{$i = \text{max\_levels}$}
    \State \Return
\EndIf
\For{code $n, k, d$ in $\text{codes}$}
    \State Append $\text{code}$ to $\text{current\_sequence}$
    \State Evaluate code: Calculate output error rate $p_i$ and memory $M_i$
    \If{$p_{i} > p_{i-1}$ or $M_{i} > M_{\text{max}}$}
        \State \textbf{Continue}: Discard this branch and continue to the next node at this level.
    \ElsIf{$\frac{p_{i}}{K_{i}} < p_{\text{target}}$}
        \State \textbf{Record and Continue}: Add sequence into $\text{Viable\_Sequences}$ and continue to the next node at this level.
    \Else
        \State \textbf{Explore Further}: $\text{DFS\_search}(i + 1, \text{current\_sequence})$
    \EndIf
\EndFor
\end{algorithmic}
\end{algorithm}

\section{Memory overhead - comparison between schemes for distributed QC}
\label{app:memory_overhead}
To estimate the memory requirements for our distillation scheme with state injection, we need to find the required distance for each surface code patch such that no logical level gate fails.
We use the relation between physical and logical errors for the rotated surface code based on simulations using Stim ~\cite{gidney2021stim}:
\begin{align}
    P_\textrm{fail}(p_\textrm{gate}) \sim 0.02  \Big( \frac{p_\textrm{gate}}{p_\textrm{gate}^\textrm{th}} \Big)^{d_e},
\end{align}
where $d_e = d/2$ for even distance, and $(d+1)/2$ for odd distance.
In calculating the patch size, we do not need to take into account the Bell pair infidelity, as it is included in the state injection output infidelity. Using $p_\textrm{gate}^\textrm{th}=1.1\%$ \cite{Stephens2014SurfaCecodeThresholds},
we find that $d=19$ is required to achieve logical error rate below $10^{-12}$. Therefore, the total number of qubits in each patch is $2d^2 -1 = 721$, and the total memory of qubits required for distillation can be estimated by the buffer size multiplied by each patch size.
To save more memory, the surface code distances may be changed during the distillation sequence such that small surface codes are used in the initial steps and large surface codes are used in the final steps. We leave this optimization for future work.

The resulting memory requirement for each scheme is presented in Table ~\ref{table:comparison_memory}. Note that for all distillation schemes, the values in each row are not monotonic with Bell pair infidelity, as they are calculated from a protocol optimized for communication overhead, which might be utilizing fewer logical qubits than the buffer size constraint.

To estimate the memory requirements for the lattice surgery scheme, we find the required patch size for each set of error rates accounting for both Bell pair errors and gate errors.
We make the optimistic assumption that, when used for lattice surgery, physical Bell pairs do not experience memory errors while they are waiting to be used.
Since there is no injection step, the network error rate is then equal to the error rate input to the distillation scheme.
We describe this calculation in \cref{app:overheads_additional_Bell_errors}.
In \cref{table:comparison_memory}, the memory overhead is taken to be equal to the patch size in each node.

\begin{table*}
    \centering
    \begin{tabular}{| c || c | c | c | c | c|}
    \hline
        \diagbox{Scheme}{Network error rate} & 0.1\% & 1\% & 5\% & 10\% & 15\%\\
     \hline \hline
         Distillation input error rate & 0.35\% & 1.25\% & 5.2\% & 10.2\% & 15.2\% \\
         \hline \hline
         BDSW-2EPP scheme - overhead  & 5,047 & 5,047 & 6,489 & 6,489 & 7,931 \\
         \hline
         BDSW-2EPP scheme with Y basis - memory overhead  & 4,326 & 4,326 & 5,768 & 5,768 & 6,489 \\
         \hline \hline
         Constant-overhead distillation (buffer = 10) - memory overhead & 7210  & 7210 & 7210 & 6489 & 7210 \\
         \hline
         Constant-overhead distillation (buffer = 30) - memory overhead & 21,630  & 20,909 & 21,630 & 21,630 & 21,630 \\
         \hline
         Constant-overhead distillation (buffer = 50) - memory overhead &  36,050 & 20,909 & 21,630 & 22,351 & 23,072 \\
         \hline
         Constant-overhead distillation (buffer = 100) - memory overhead & 70,658  & 62,006 & 72,100 & 64,169 & 23,072 \\
         \hline\hline
         Lattice surgery - memory overhead & 1,089  & 1,369 & 5,329 & 22,201 & 142,129 \\
         \hline \hline
    \end{tabular}
    \caption{Physical network memory requirements for different schemes for distributed quantum computation, considering encoding single logical qubits on surface codes. Lattice surgery and state injection values are estimated according to gate error rate $p_{\text{gate}}=0.1\%$ and target output error rate $p_{\text{target}} = 10^{-12}$}.
    \label{table:comparison_memory}
\end{table*}

\section{In-Depth Overview of the Distillation Pipeline}
\label{app:Distillation_pipeline_in_depth_overview}

In this section, we present a detailed analysis of the space-time cost of our entanglement distillation procedure, this time taking full account of considerations such as throughput, memory usage throughout the procedure, and parallelism.
To minimize the space-time cost and achieve good distillation throughput, we structure the distillation process as a pipeline, in which later stages of distillation are performed in parallel with further repetitions of earlier distillation stages.
This yields a slightly larger memory footprint for distillation, but reduces the overall space-time cost and ensures that all computation resources are well utilized at all times.

Our analysis here focuses primarily on the average throughput, neglecting more detailed scheduling and handling of distillation failure events and intermediate storage needed to ensure that each stage can start without waiting.
We leave detailed analysis of these aspects to future work, as they will likely be highly dependent on details of the quantum computing system.
However, we expect our conclusions to be accurate to within a small constant factor.

\subsection{Parallel Unencoding Circuit}
As shown in Fig.~\ref{fig:generic_encoding}, the QED distillation protocol consists of stabilizer projection and post-selection on the patterns on the two sides being identical, followed by unencoding.
There are multiple ways that this procedure could be performed.
Here, we focus on the basic protocol of directly running the unencoding circuit, which simultaneously unencodes stabilizers onto single qubits in a deterministic state, while also unencoding logical qubits of the distillation code into single qubits and logical Bell pairs of the distillation code into individual Bell pairs (still encoded in some inner QEC code).
We perform the unencoding circuit in place to minimize the space cost.
Other approaches with different space-time trade-offs exist, and we briefly comment on them at the end of this section.
\begin{figure}
\centering
\includegraphics[width=0.8\columnwidth]{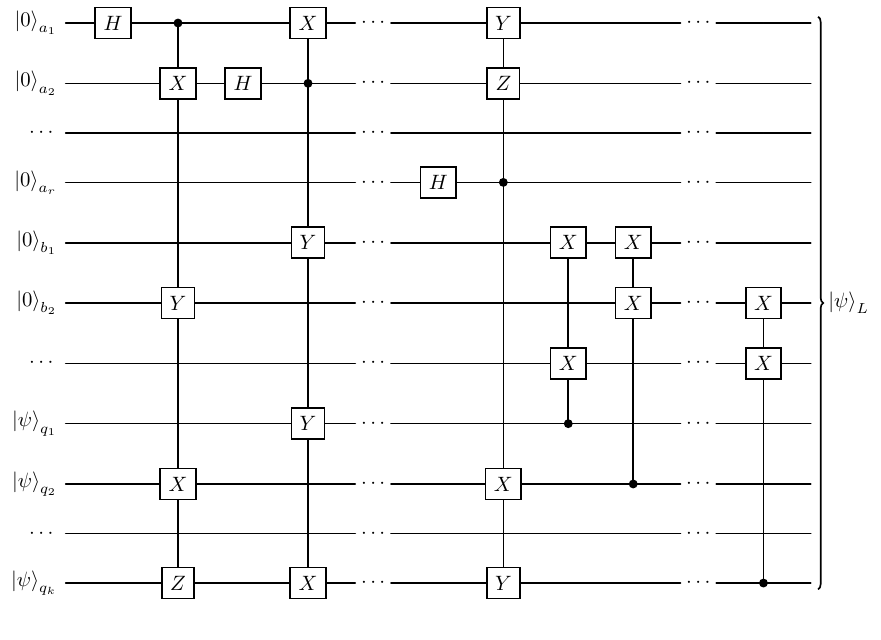}
\caption{Illustration of the structure of a generic encoding circuit and its parallelized realization in $3n-3$ two-qubit gate layers. Here, the stabilizers and logical operators have been put in a canonical form following Ref.~\cite{cleve1997efficient}, where the first $r$ qubits (labeled $a$) correspond to stabilizers that have $X$ support, the next $n-k-r$ qubits (labeled $b$) correspond to stabilizers that only have $Z$ support, and the last $k$ qubits correspond to the physical qubit state $\ket{\psi}$ that is encoded. The output is a logical encoded state $\ket{\psi}_L$. Running this circuit in reverse performs the desired unencoding.
\label{fig:generic_encoding}}
\end{figure}

The unencoding circuit we choose to run is a parallelized version of the encoding circuit described in Ref.~\cite{cleve1997efficient}, run in reverse, see Fig.~\ref{fig:generic_encoding}.
In this construction, one first puts the stabilizers and $\bar{X}$ logical representatives in a suitable canonical form.
Suppose that in this canonical form, there are $r_X$ stabilizers that have nontrivial $X$ support on some qubits.
One then applies $I+M_i$ for each of the stabilizers $M_i$ ($i=1,2,...,r_X$), which is achieved by a Hadamard gate on the $i$th qubit followed by controlled operations from the $i$th qubit to other qubits.
This is followed by controlled operations from each of the last $k$ qubits to the first $n-k$ qubits, encoding each of the basis states of those logical qubits into the corresponding logical operators.
The remaining $Z$-only stabilizers are automatically satisfied for the given initial state and commute through the circuit, thereby remaining valid stabilizers.

Running this circuit in reverse, the unencoding circuit has the following structure: in the $i$th layer, perform operations controlled by the $i$th qubit and targeting other qubits, with the middle $n-k-r_X$ layers corresponding to $Z$-only stabilizers trivial.
For simplicity, we assume that the Hadamard operations are merged into the two qubit gate when counting the number of gate layers, although we remark that this will depend on the detailed error correction scheme employed for the inner logical qubits.
The single-control, multi-target operations can be efficiently performed in a constant number of timesteps using lattice surgery~\cite{fowler2018low} or transversal CNOTs and a GHZ state.
To reduce space usage, we instead decompose it into individual two-qubit gate layers in-place with the following strategy, as is also illustrated in Fig.~\ref{fig:generic_encoding}.
Here, we denote a two-qubit gate controlled by qubit $i$ and with qubit $j$ as target as $C_{i,j}$, regardless of the basis.

\begin{enumerate}
\item[$1$.] Perform gate $C_{1,2}$.
\item[$2$.] Perform gate $C_{1,3}$.
\item[$3$.] Perform gates $C_{1,4}$ and $C_{2,3}$.
\item[$4$.] Perform gates $C_{1,5}$ and $C_{2,4}$.
\item[$5$.] Perform gates $C_{1,6}$, $C_{2,5}$, and $C_{3,4}$.
\item[...] Continue.
\item[$n-1$.] Perform gates $C_{1,n}$, $C_{2,n-1}$, ...
\item[$n$.] Perform gates $C_{2,n}$, $C_{3,n-1}$, ...
\item[$n+1$.] Perform gates $C_{2,1}$, $C_{3,n}$, ...
\item[$n+2$.] Perform gates $C_{3,1}$, $C_{4,n}$, ...
\item[...] Continue.
\item[$2n-1$.] Perform gates $C_{n,1}$, $C_{n-1,2}$, ...
\item[...] Continue.
\item[$3n-5$.] Perform gates $C_{n,n-3}$, $C_{n-1,n-2}$.
\item[$3n-4$.] Perform gate $C_{n,n-2}$.
\item[$3n-3$.] Perform gate $C_{n,n-1}$.
\end{enumerate}

It is easy to verify that this scheduling performs all desired gate operations, and in the same order as required for each qubit.
Since not all gates in the circuit are needed, this circuit depth is only an upper bound.
Inspecting Fig.~\ref{fig:generic_encoding}, we see that there are no gates between the last $k$ qubits, so the last gate we need to perform is $C_{n,n-k}$, which results in a total gate depth of $3n-2-k$.
We thus assume that the distillation procedure outputs all Bell pairs in time $(3n-2-k)T_{\text{gate}}$, regardless of the success or failure of the procedure.
Here, $T_{\text{gate}}$ is the logical gate speed.
For a given code, it may be possible to further reduce the number of layers required although we do not use this in our circuit size models.
For example, in Fig.~\ref{fig:parity_code_unencoding}, we illustrate the unencoding circuit for the $[[n,n-2,2]]$ quantum parity code, which only requires $n-2$ two-qubit gate layers for even $n$.

As a final note, we briefly comment on other methods to perform the distillation step, whose detailed analysis we leave for future work.
Instead of performing the unitary un-encoding circuit in-place in linear depth, we can use various methods for space-time trade-offs~\cite{fowler2012timeoptimal} to perform these operations in lower depth.
For example, using an ancilla GHZ state, we can apply single-control, multi-target operations~\cite{fowler2018low}, which dis-entangles the stabilizer qubits in constant depth.
The measurement results of qubits that correspond to stabilizers allow for error detection, while the unencoded logical Bell pairs, upon heralded success, can be used for further layers of distillation.
We leave more careful analysis of the scheduling and space-time tradeoffs, together with bottlenecks due to the probabilistic nature of the distillation and entanglement distribution process, to future research.

\begin{figure}
\centering
\includegraphics[width=0.5\columnwidth]{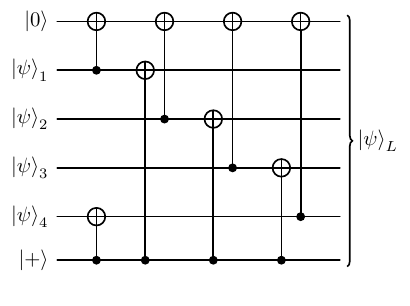}
\caption{Illustration of an (un)-encoding circuit for the [[6,4,2]] quantum parity code. The same structure can be easily generalized to any quantum parity code with an even number of qubits.}
\label{fig:parity_code_unencoding}
\end{figure}

\subsection{Space-Time Cost of a Single Distillation Stage}
We now consider the throughput and memory footprint of a single distillation stage consisting of a $[[n,k,d]]$ quantum code.
In practice, we may also be interested in the latency, i.e. the time it takes for the first high-fidelity logical Bell pair to be output from the process.
However, it is the throughput that ultimately determines the rate at which high-fidelity logical Bell pairs are produced, so we focus on this quantity.
As inputs to the analysis, we have the time between arrival of consecutive inputs $T_{\text{in}}$, the time it takes to perform the unencoding step (distillation step) $T_{\text{dis}}$, and the maximum logical memory size $B$ allocated to this stage.
Analysis of the distillation procedure then provides an estimate of $T_{\text{out}}$, the average time between consecutive outputs of the distillation procedure.
Because we are focusing on the throughput here, both $T_{\text{in}}$ and $T_{\text{out}}$ are defined as the average time \textit{per logical Bell pair}.

With these assumptions, different situations can arise (Fig.~\ref{injection_pipeline}), depending on what resource is the primary bottleneck.
\begin{itemize}
\item Slow operations ($T_{\text{dis}} >  nT_{\text{in}}$) and large space ($B \geq n\lceil \frac{T_{\text{dis}}}{nT_{\text{in}}} \rceil$), or fast operations ($T_{\text{dis}} <  nT_{\text{in}}$): In this case, the throughput of the distillation is sufficiently large, and we are bottlenecked by the average time between inputs $T_{\text{in}}$. The average time between outputs is given by
\begin{align}
T_{\text{out}}=\frac{nT_{\text{in}}}{(1-p_{\text{fail}})k}.
\end{align}
\item Slow operations ($T_{\text{dis}} >  nT_{\text{in}}$) and limited space ($B \leq n\lceil \frac{T_{\text{dis}}}{nT_{\text{in}}} \rceil$): In this case, we are bottlenecked instead by the distillation step, and because the distillation cannot keep up with the input rate, we will be discarding some fraction of the inputs in the steady state. The average time bewteen outputs is given by 
\begin{align}
T_{\text{out}}=
\left\lceil \frac{B}{n_{\text{total}}}\right\rceil\frac{T_{\text{dis}}}{(1-p_{\text{fail}})k},
\end{align}
since we can perform $B/n_{\text{total}}$ copies of distillation in parallel, each taking time $T_{\text{dis}}$ and outputting $k$ Bell pairs, and failures in the distillation will cause us to need to retry.
$n_{\text{total}}$ also accounts for possible ancilla qubits used to help with the distillation process.
\end{itemize}

For a typical distillation stage, as discussed above, we can estimate $T_{\text{dis}}=(3n-2-k)T_{\text{gate}}$.
The concatenated distillation process should ensure that earlier stages of distillation are not going to waste; therefore, the throughput of earlier stages of the pipeline should not exceed that of later stages, and ideally the throughput in all stages are matched.
For classical codes we calculate $p_{\text{fail}}$ exactly, and for quantum codes, we take the bound $1 - p_{\text{fail}} \geq (1-p_{\text{in}})^{n}$ discussed above.

Although our analysis here was formulated in terms of a single stage of distillation, we can apply the same considerations to the injection of physical Bell pairs into logical Bell pairs, setting $T_{\text{in}}=T_{\text{Bell}}$, the average time separation between consecutive physical Bell pairs (we ignore possible stochasticity in this process), $T_{\text{dis}}=T_{\text{gate}}$, assuming that state injection takes a similar amount of time as a single logical gate, which is true to within a small constant factor for both lattice surgery~\cite{fowler2018low} and transversal gate~\cite{zhou2024algorithmic} architectures.
Each state injection requires a space footprint of a single logical qubit and we set $n=1$ above.

\begin{figure}[h]
  \centering
  \includegraphics[width=\columnwidth]{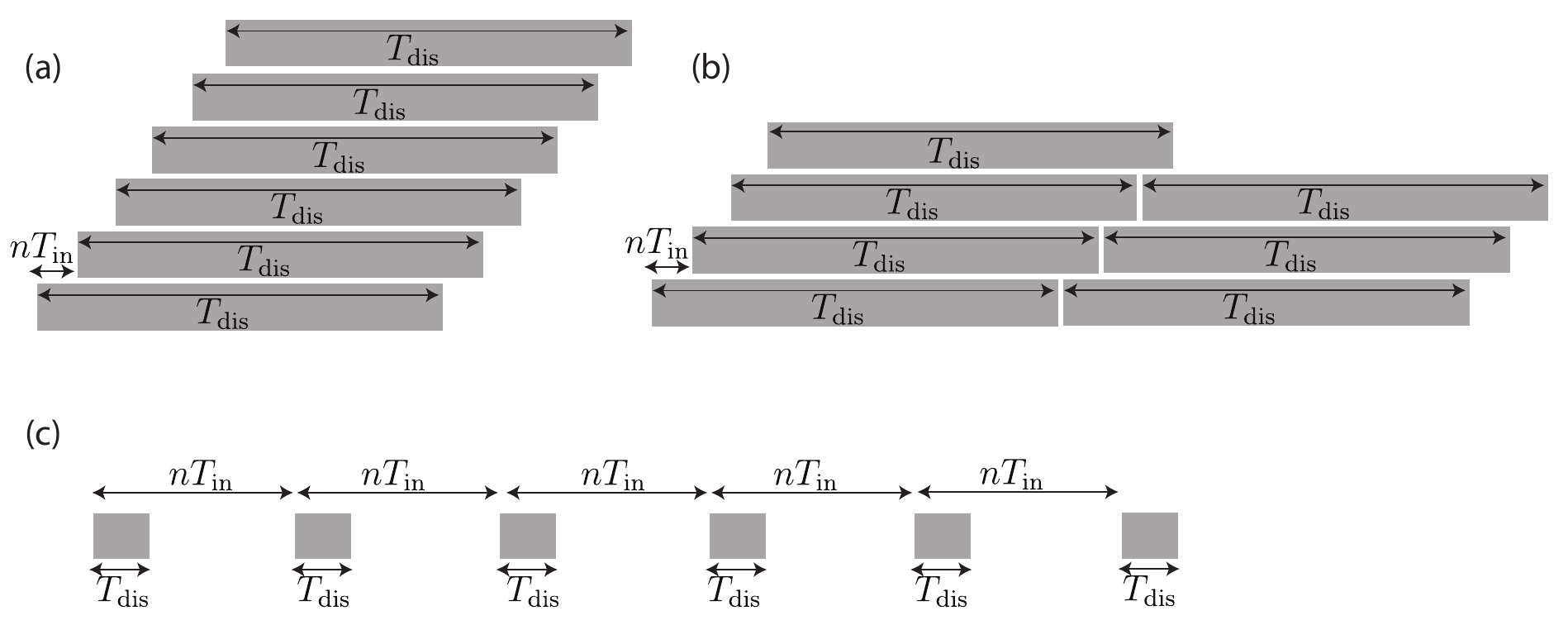}
  \caption{Illustration of the distillation throughput in different cases. (a) Slow operations, large space. (b) Slow operations, small space. (c) Fast operations or slow Bell pair rate. In (a,c), there is sufficient workspace to process all input Bell pairs, and the process is bottlenecked by the input; in (b), the distillation process cannot keep up with the input rate, and some input Bell pairs will be discarded in steady state.
  \label{injection_pipeline}
  }
\end{figure}

\subsection{Space-Time Cost of Pipelined Distillation}
In this section, we estimate the space-time cost of a distillation pipeline.
We focus on the case where the distillation throughput is limited by the incoming rate of Bell pairs.
This is likely the most relevant regime for several reasons: First, as discussed in the main text, current and projected technologies tend to have relatively slow internode entanglement distribution speeds, compared to local operation speeds; Second, later stages of distillation in our high-rate distillation scheme will naturally have multiple copies executing in parallel to avoid correlated errors from earlier stages (Fig.~\ref{fig:concat-distillation-process}), and therefore are unlikely to be the bottleneck.

For state injection, if the input Bell pair rate is slow, then we can inject them one at a time, and only need space corresponding to a single logical qubit.
Otherwise, we can parallelize so that the throughput of state injection keeps up with the input Bell pair rate.
The combined space usage of this stage is thus
\begin{align}
B_0=\max\left(\left\lceil \frac{T_{\text{inject}}}{T_{\text{Bell}}}\right\rceil, 1\right)=\left\lceil \frac{T_{\text{inject}}}{T_{\text{Bell}}}\right\rceil,
\end{align}
and we produce a logical Bell pair in time $T_{\text{Bell}}/(1-p_{\text{fail},0})$ on average.
Here, $p_{\text{fail},0}$ is the success rate of Bell state injection, and the state injection takes an amount of time comparable to the logical gate time $T_{\text{inject}}\approx T_{\text{gate}}$, which is the case for both common schemes based on lattice surgery and recent schemes based on transversal gates and algorithmic fault tolerance~\cite{li2015magic,fowler2018low,zhou2024algorithmic}.

At subsequent stages of distillation, a single copy of the distillation can take as input $n_i$ qubits and operates on it for time $(3n_i-2-k_i)T_{\text{gate}}$ to produce an output, with some failure probability $p_{\text{fail},i}$.
In order to process all the inputs from previous stages, we will need to have enough space to store the $K_{i-1}=\prod_{j=1}^{i-1}k_j$ copies that need to be distilled in parallel.
The time required to prepare the inputs, based on the available number of input physical Bell pairs and distillation rates at each level, is given by
\begin{align}
T_{\text{input},i}=n_i\prod_{j=0}^{i-1} \frac{n_j}{1-p_{\text{fail},j}},
\end{align}
with $n_0=1$.
At higher levels in the distillation, it becomes increasingly likely that the physical Bell pair production rate becomes the bottleneck.
In either case, we can estimate the number of qubits required as 
\begin{align}
B_i=\qty(\frac{(3n_i-2-k_i)T_{\text{gate}}}{T_{\text{input},i}}+1)n_i K_{i-1}.
\end{align}
This accounts for both the possible parallelization of distillation across multiple copies when the input time is shorter than the distillation time
\begin{align}
T_{\text{input},i}\leq (3n_i-2-k_i)T_{\text{gate}},
\end{align}
as well as the amount of time it takes for the workspace to become available again for accumulating qubits following a round of distillation, see also Fig.~\ref{injection_pipeline}(a,c).
The total space usage is then the sum of the space usage at all stages
\begin{align}
B_{\text{all}}=\sum_{i=0}^l B_i,
\end{align}
and we produce $\prod_{j=1}^l k_l$ Bell states every $T_{\text{Bell}}\prod_{j=0}^l\frac{n_i}{1-p_{\text{fail},i}}$.

Our analysis here has focused on the regime where the throughput is limited by the input physical Bell pair rate.
However, the analysis can be straightforwardly generalized to other cases, such as those bottlenecked by memory constraints and the throughput of the un-encoding step.
In this case, we can try to construct a well-matched pipeline, in which the throughput of each distillation stage is equal, as illustrated in Fig.~\ref{fig:example_pipeline_uniform}.

\begin{figure}[h]
\centering
\includegraphics[width=0.9\columnwidth]{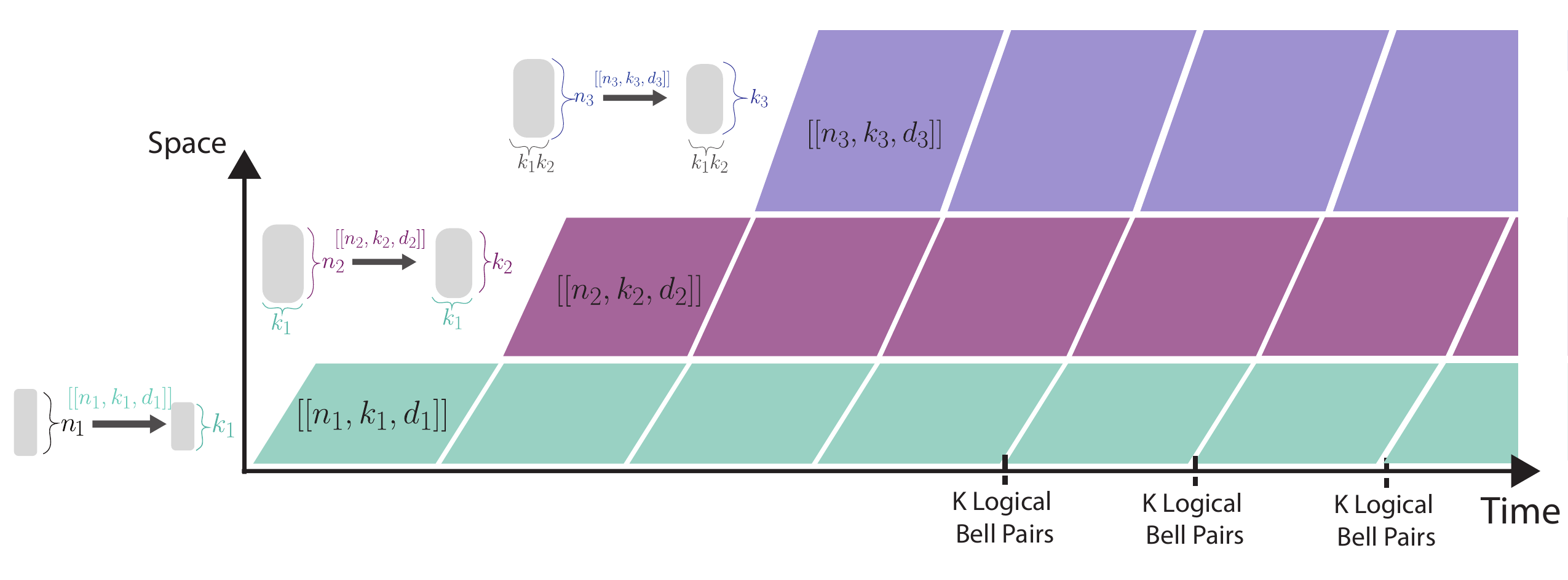}
\caption{A well-matched pipeline in which the throughput of each stage is the same, thereby ensuring that no individual stage is a bottleneck for the whole process.
\label{fig:example_pipeline_uniform}
}
\end{figure}

\section{Small Buffer Sizes in Distillation Sequences}
\label{app:small_buffer}
While our study primarily details distillation schemes with buffer sizes of 30 or more logical qubits, it is noteworthy that efficient schemes exist for smaller buffer sizes as well.
As an extreme example, Ref.~\cite{gidney2023tetrationally} shows that very high fidelity entanglement distillation is possible with very low buffer size, at the cost of overhead in the number of Bell pairs used.

Consider a scenario with a Bell pair error rate of 1\%. In such a case, an optimized distillation sequence can be formulated even with a buffer size as small as 8 qubits. This optimized sequence is as follows:
\begin{equation}
    [3,1,3]_X \rightarrow [2,1,2]_Y \rightarrow [2,1,2]_X \rightarrow [[4,2,2]].
\end{equation}
Here, the total memory required is 8 logical qubits used for the network on each node. The overhead for generating one logical gate is 25 Bell pairs, allowing for relatively high rates of distributed logical gates in systems constrained by smaller buffer sizes.

\section{Constant-overhead distillation schemes additional results}
\label{app:overheads_additional_Bell_errors}

\begin{figure}[h]
  \centering
  \includegraphics[width=0.8\columnwidth]{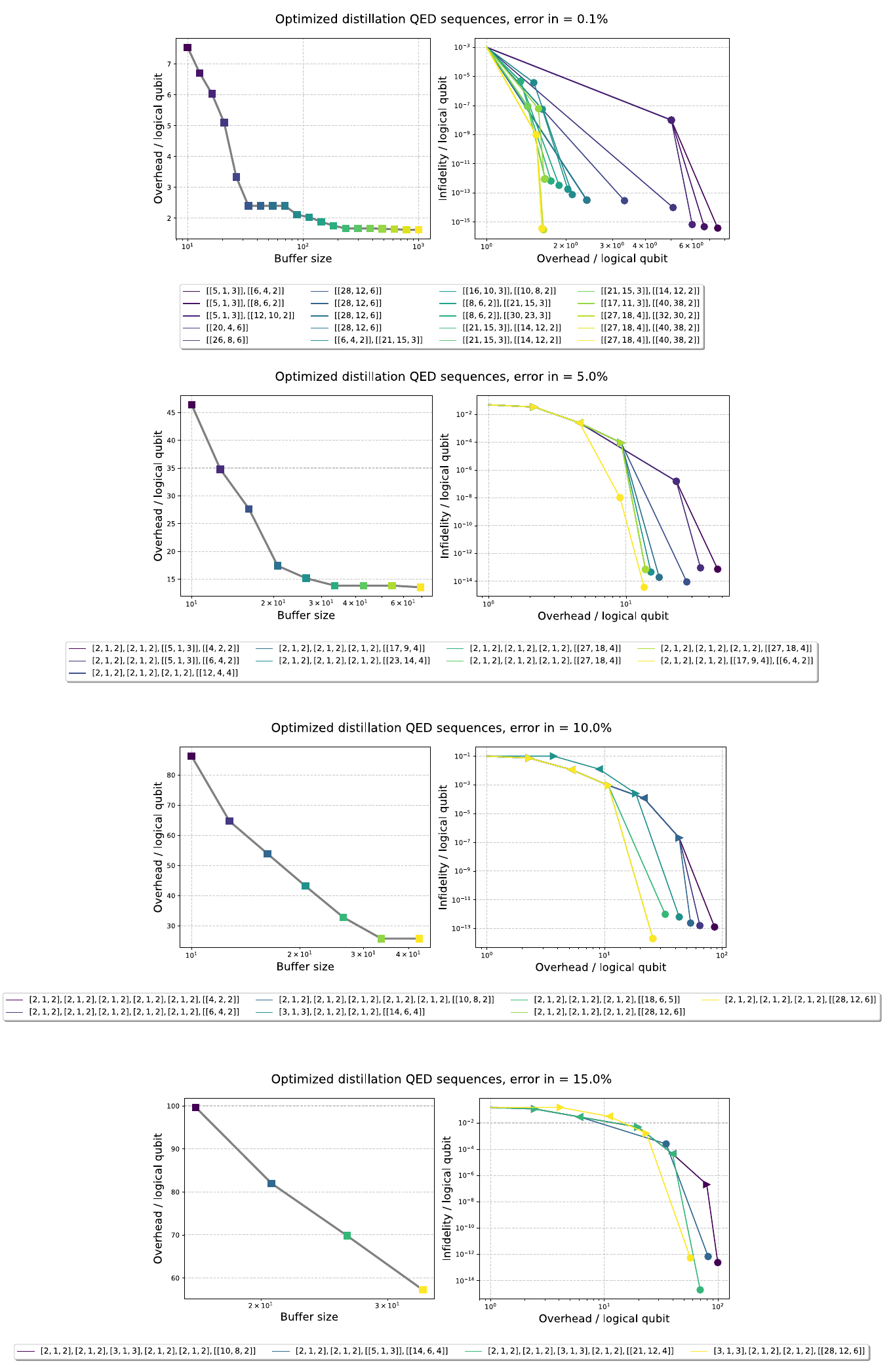}
  \caption{Optimized distillation Quantum Error Detection (QED) sequences at various Bell pair error rates. The x-axis represents the buffer size, while the y-axis shows the overhead per logical qubit. Various points on the plot correspond to different optimized sequence configurations, each uniquely identified and color-coded for clarity. Quantum codes are represented by circles, and classical codes are represented by triangles, pointing upwards, to the right, or left, corresponding to the Z, X, and Y bases of the classical codes, respectively.
  \label{sequence_p_in_many}
  }
\end{figure}

\section{Lattice Surgery Quantitative Analysis}
\label{app:lattice_surgery}
For analysis of the lattice surgery overhead, we adopt the model from \cite{ramette2023fault, sinclair2024fault} which omits single qubit errors. The logical error rate is taken to be
\begin{align}
    P_\textrm{fail}(p_\textrm{bulk},p_\textrm{boundary}) \approx 0.03\left( \frac{p_\textrm{bulk}}{p_\textrm{bulk}^\textrm{th}} \right)^{d_e} + 0.03\left( \frac{p_\textrm{boundary}}{p_\textrm{boundary}^\textrm{th}} \right)^{d_e},
\end{align}
where \(d_e\) is the effective distance of the code: \(d/2\) for even distances and \((d+1)/2\) for odd distances. The phenomenological thresholds are \(p_\textrm{boundary}^\textrm{th} = 10\%\), \(p_\textrm{bulk}^\textrm{th} = 3\%\). The prefactor 0.03 is accurate for subthreshold error rates as detailed in \cite{ramette2023fault}. However, for Bell pair error rates exceeding \(1\%\), the prefactor increases, rendering our Bell pair overhead estimation as merely a lower bound, and further distillation is necessary before lattice surgery to reduce Bell pair errors.

The bulk and boundary noise for the surface code is modeled using gate and Bell pair error rates \cite{ramette2023fault, sinclair2024fault}:
\begin{equation}
    \begin{split}
        p_{\text{bulk}} &= 2 p_{\text{gate}}, \\
        p_{\text{boundary}} &= 2.5 p_{\text{gate}} + 0.5 p_{\text{Bell}}.
    \end{split}
\end{equation}

Using the above, we calculate the necessary code distance to achieve a logical error rate below \(10^{-12}\), and thereby estimate a lower bound for the physical Bell pair overhead per logical gate as \(2(d^{2}+(d-1)^{2})-1\). Here, the standard (not rotated) surface code is used due to its favorable boundary conditions for distributing the gates required for stabilizer measurements~\cite{ramette2023fault}. The calculated overheads for different Bell pair error rates are as follows:
\begin{table}[h]
\begin{center}
\begin{tabular}{|c|c|c|}\hline
  \(p_{\text{Bell}}\) & \(d\) & Bell pairs \\\hline\hline
  \(0.1\%\) & \(17\)  &  \(1089\)\\\hline
  \(1\%\)   & \(19\)  &  \(1369\)\\\hline
  \(5\%\)   & \(37\)  &  \(5329\)\\\hline
  \(10\%\)  & \(75\)  &  \(22,201\)\\\hline
  \(15\%\)  & \(189\) &  \(142,129\)\\\hline
\end{tabular}
\end{center}
\end{table}

\end{document}